\documentclass[a4paper,11pt]{article}
\usepackage{amsmath}
\usepackage{amsfonts,amssymb}
\usepackage{eucal}
\usepackage{amsthm}
\usepackage[dvipdfmx]{graphicx} 
\numberwithin{equation}{section}
\usepackage{url}

\newcommand{\simgt}{\lower.5ex\hbox{$\; \buildrel > \over \sim \;$}}
\newcommand{\simlt}{\lower.5ex\hbox{$\; \buildrel < \over \sim \;$}}

\def\re{\mathbb{R}}
\def\co{\mathbb{C}}

\renewcommand{\Re}{\text{{\rm Re}}}
\renewcommand{\Im}{\text{{\rm Im}}}

\renewcommand{\ker}{\text{\rm Ker}}

\newtheorem{thm}{Theorem}
\newtheorem{lem}[thm]{Lemma}
\newtheorem{prop}[thm]{Proposition}
\newtheorem{cor}[thm]{Corollary}

\theoremstyle{definition}

\theoremstyle{remark}
\newtheorem{rem}{Remark}

\title{Schr\"odinger operators with concentric $\delta$--shell interactions}

\author{%
Masahiro Kaminaga\thanks{Corresponding author. Email: kaminaga@mail.tohoku-gakuin.ac.jp}\\
{\small Department of \textit{Information Technology}, Faculty of Engineering,}\\
{\small Tohoku Gakuin University, Sendai, Japan}\\
}

\date{}

\begin{document}
\maketitle

\begin{abstract}
We study Schr\"odinger operators on $\mathbb R^3$ with finitely many concentric spherical $\delta$--shell interactions.
The operators are defined via quadratic forms and characterized by continuity across each shell together with the standard jump condition for the normal derivative.
Using a boundary integral approach based on the free Green kernel and single--layer potentials, we derive an explicit resolvent formula for an arbitrary number of shells with bounded coupling strengths.
This yields a concrete Kre\u{\i}n--type representation and a boundary operator whose noninvertibility characterizes the discrete spectrum, and it is compatible with partial--wave reduction under rotational symmetry.

We then specialize to the two--shell case with constant couplings and obtain a detailed description of the negative spectrum.
In particular, we prove that the ground state, when it exists, lies in the $s$--wave sector and derive an explicit secular equation for bound states.
For large shell separation, each bound level approaches the corresponding single--shell level with exponentially small corrections, while a genuine tunneling splitting appears when the single--shell levels are tuned to coincide.
As a simple calibration, we relate the two--shell parameters to representative core--shell quantum dot scales and identify the qualitative distinction between Type~I and Type~II configurations.

\medskip
\noindent\textbf{Keywords.}
$\delta$--shell interaction; singular perturbation; resolvent formula; Kre\u{\i}n--type formula; partial--wave decomposition; tunneling splitting

\medskip
\noindent\textbf{Mathematics Subject Classification (2020).}
35J10; 35P15; 47A10; 47F05
\end{abstract}

\section{Introduction}\label{sec:intro}

Surface--supported singular perturbations of the Laplacian 
provide a standard class of solvable Schr\"odinger operators. 
A general extension--theoretic framework for solvable models,
including point interactions, can be found in
Albeverio et al.~\cite{AGHH}.
For singular interactions supported on hypersurfaces,
boundary integral and layer potential approaches were developed by
Brasche--Exner--Kuperin--\v{S}eba~\cite{BrascheExnerKuperinSeba1994},
including an abstract Kre\u{\i}n--type resolvent formula in a general
measure--valued framework.

For radially symmetric penetrable wall models, Ikebe and Shimada~\cite{IkebeShimada1991}
studied the spectral and scattering theory in detail, and Shimada obtained
further results on approximation by short--range Hamiltonians~\cite{Shimada1992},
low--energy scattering~\cite{Shimada1994LowE}, and analytic continuation of the
scattering kernel~\cite{Shimada1994AnalCont}.
More recently, essential self--adjointness for systems with several (possibly
infinitely many) concentric shells was proved by Albeverio, Kostenko, Malamud,
and Neidhardt~\cite{AlbeverioKostenkoMalamudNeidhardt}.

The rigorous treatment of $\delta$--interactions supported on a sphere
was initiated by Antoine, Gesztesy, and Shabani~\cite{AntoineGesztesyShabani1987},
where the single--shell case was analyzed within an extension--theoretic framework.
In the specific setting of finitely many concentric spherical shells with
$\delta$--type interactions, Shabani~\cite{Shabani1988JMP} investigated the model
by an explicit reduction to radial one--dimensional problems.
In particular, for constant shell strengths the matching conditions at the
interfaces lead to a finite--dimensional determinant condition in each partial wave,
so that the bound state problem can be treated in a fully explicit manner.
These works provide closely related precedents for the present study,
as they address essentially the same geometric configuration and exhibit,
already at the level of partial--wave reduction, the characteristic coupling
mechanism induced by multiple shells.
While Shabani's analysis proceeds via a reduction to one--dimensional radial ODEs and finite--dimensional matching conditions in each partial wave (in particular for constant shell strengths),
our approach keeps the problem in three dimensions and yields a boundary--integral resolvent representation that remains valid for arbitrary $N$ and nonconstant surface strengths $\alpha_j\in L^\infty(S_j;\mathbb R)$.

Our approach is complementary.
Rather than revisiting the operator--theoretic construction,
we allow nonconstant strengths $\alpha_j\in L^\infty(S_j;\mathbb R)$ for arbitrary $N$
and derive a concrete boundary integral Kre\u{\i}n--type resolvent formula
directly in terms of the free Green kernel and single--layer potentials.
In the constant two--shell case we then provide a detailed large--separation
analysis, including the tuned tunneling splitting on the scale
$e^{-\kappa_0 d}$.
This formulation makes explicit the connection between the boundary integral picture and the partial--wave secular equations under rotational symmetry.
Related developments on the scattering side for finitely many concentric sphere
interactions were later obtained by Hounkonnou, Hounkpe, and
Shabani~\cite{HounkonnouHounkpeShabani1997}.

Building on these works, we consider Schr\"odinger operators on $\mathbb R^3$
with finitely many concentric spherical $\delta$--shell interactions.
Let $0<R_1<R_2<\cdots<R_N$ and set
$S_j:=\{x\in\mathbb R^3:\ |x|=R_j\}$ for $j=1, 2,\ldots, N$.
For each $j$ we allow a bounded measurable surface strength
$\alpha_j\in L^\infty(S_j;\mathbb R)$.
Formally, the model is
\begin{equation}\label{eq:N-shell}
H_N
=
-\Delta
+
\sum_{j=1}^{N}\alpha_j\,\delta(|x|-R_j),
\end{equation}
where $\alpha_j$ acts by multiplication on the trace of a function on $S_j$.
We define $H_N$ rigorously by the quadratic form method.
At an abstract level, we note that resolvent representations for self--adjoint extensions are well-known.
For a wide class of singular perturbations, we can express $(H_N - z)^{-1}$ in terms of the free resolvent and an operator--valued boundary term of Kre\u{\i}n type.
For example, we refer to the extension theoretic framework in \cite{AGHH} and to related formulations for singular interactions in \cite{BrascheExnerKuperinSeba1994}.
See also \cite{Posilicano2001} for Kre\u{\i}n--type resolvent formulas in a general
singular perturbation framework.
For $\delta$-- and $\delta'$--interactions supported on hypersurfaces, including
analogues of the Birman--Schwinger principle and a variant of Kre\u{\i}n's formula,
we refer to \cite{BehrndtLangerLotoreichik2013}.
The main result of this paper is an explicit boundary integral representation 
for concentric spherical $\delta$--shell interactions, 
written directly in terms of the free Green kernel and single--layer potentials, 
without invoking boundary triple machinery. 

More precisely, we derive an explicit resolvent formula for $H_N$ in terms of
layer potentials on the shells.
Here $S^2=\{\omega\in\mathbb R^3:\ |\omega|=1\}$ denotes the unit sphere, and
this representation yields a boundary operator $K_N(z)$ on
$\bigoplus_{j=1}^N L^2(S^2)$
such that, for $z\in\mathbb C\setminus[0,\infty)$,
the eigenvalues of $H_N$ in $\mathbb C\setminus[0,\infty)$ are characterized by
the noninvertibility of $K_N(z)$.

We then relate this representation to the partial--wave reduction 
and provide a detailed spectral analysis in the two--shell case.
After establishing the $N$--shell resolvent framework, we specialize to the
two--shell case $N=2$,
\begin{equation}\label{eq:two-shell}
H
=
-\Delta
+
\alpha_1\,\delta(|x|-R_1)
+
\alpha_2\,\delta(|x|-R_2),
\qquad
0<R_1<R_2,
\end{equation}
which already exhibits a nontrivial coupling mechanism between distinct
interfaces.
For explicit closed formulas and a partial--wave reduction, we further restrict,
from that point on, to the rotationally symmetric setting where $\alpha_1$ and
$\alpha_2$ are constants.
In this case we obtain a detailed description of the negative spectrum,
including a detailed description of the $s$--wave eigenvalues $E=-\kappa^2<0$.
Depending on the signs of $\alpha_1$ and $\alpha_2$, the $s$--wave sector
supports zero, one, or two bound states.
In the near--decoupling regime of large shell separation $d=R_2-R_1$, each bound
level is close to the corresponding single--shell level.
Moreover, when the parameters are tuned so that the two single--shell levels
coincide in this limit, say at a common energy $E_0=-\kappa_0^2<0$, where
$\kappa_0>0$ denotes the decay rate of the corresponding one--shell bound
state, the resulting pair of eigenvalues exhibits a tunneling splitting with
an exponentially small gap of order $e^{-\kappa_0 d}$.

A motivation for this analysis comes from semiconductor core--shell nanocrystals.
The first observation of size quantization was reported by Ekimov and
Onushchenko~\cite{EkimovOnushchenko1981}, followed by the effective--mass theory
of Efros and Efros~\cite{EfrosEfros1982} and the chemical studies of Rossetti,
Nakahara, and Brus~\cite{Brus1983}.
Type~I CdSe/ZnS core--shell nanocrystals~\cite{HinesGuyotSionnest1996} realize a
structure in which both carriers remain confined in the core, while Type~II
systems with spatially separated electrons and holes were realized later and
show distinct excitonic behavior~\cite{Kim2003,OronKazesBanin2007}.
In the effective--mass picture~\cite{EfrosEfros1982}, confinement is produced by
band offsets at heterointerfaces, and our $\delta$--shell idealization reflects
this mechanism through the signs and positions of the interface parameters.

A related viewpoint appears in the tunneling theory of BenDaniel and Duke~\cite{BenDanielDuke},
where abrupt changes of the effective mass and band edge across a heterojunction
are encoded in the BenDaniel--Duke matching conditions.
In the effective--mass literature (see, for example, \cite{HarrisonQWWD4}), these
conditions are regarded as standard interface rules for very thin heterojunctions.
When the intrinsic width of an interface is much smaller than the de~Broglie
wavelength, its effect can be compressed into a surface term, and $\delta$--type
interactions provide a natural effective description.
In both the BenDaniel--Duke
model and the present $\delta$--shell idealiza\-tion, the wave function is continuous
across the interface, while the normal derivative may jump. In the former
case this jump is caused by a change of effective mass, whereas in the latter
it is produced by a singular interface potential. In the present work we keep
a constant effective mass and absorb interfacial microphysics into effective
surface strengths $\alpha_j$.
Since the bulk potential is taken to be zero away from the shells, a purely
repulsive interface ($\alpha>0$) does not by itself produce $L^2$ bound states in this idealization,
and realistic confinement is primarily due to finite band offsets (finite wells/barriers) and finite interface widths.
Accordingly, the discrete eigenvalues analyzed below should be understood as an idealized limit capturing the relevant scales/trends (and tunneling splittings) of effective, possibly quasi--bound, levels rather than as quantitative optical transition energies.

For the rotationally symmetric two--shell setting with constant surface
strengths, one can distinguish the Type~I and Type~II sign patterns familiar
from core--shell quantum dots.
In this idealized case, Type~I systems correspond to $\alpha_1<0<\alpha_2$
\cite{HinesGuyotSionnest1996}, while Type~II systems correspond to
$\alpha_1>0>\alpha_2$ \cite{OronKazesBanin2007}.
This correspondence is meant only in the constant--strength setting and at the level of sign patterns.
Throughout this paper, the Type~I/II terminology is used only in the constant
case of $\alpha_j$.

The paper is organized as follows.
In Section~2 we construct $H_N$ by quadratic forms and establish basic bounds.
In Section~3 we derive the boundary integral resolvent formula and the
eigenvalue condition in terms of $K_N(z)$.
In Sections~4--7 we specialize to two shells with constant couplings, prove
that the ground state is an $s$--wave, and analyze bound states, level splitting,
and tunneling effects for large separation.
Appendix~A contains a brief order--of--magnitude calibration for representative
core--shell quantum dot scales, and Appendix~B collects explicit partial--wave
matrices used in the rotationally symmetric setting.


\section{Form method and operator domain}\label{sec:form}

We fix an integer $N\ge1$ and radii $0<R_1<\cdots<R_N$.
For each $j\in\{1,\dots,N\}$ we set
$$
S_j:=\{x\in\re^3:|x|=R_j\}.
$$
Let $\alpha_j$ be a real-valued function on $S_j$ such that
$\alpha_j\in L^\infty(S_j)$.
We consider the formal Schr\"odinger operator
\begin{equation}\label{eq:HN-formal}
H_N=-\Delta+\sum_{j=1}^N \alpha_j\,\delta(|x|-R_j).
\end{equation}

For $j=1,\dots,N$, let $d\sigma_j$ denote the surface measure on
$S_j:=\{x\in\re^3:\ |x|=R_j\}$.
For $u,v\in H^1(\re^3)$ we define the quadratic form
\begin{equation}\label{eq:qN}
h_N[u,v]
=
\int_{\re^3}\nabla u\cdot\nabla v\,dx
+
\sum_{j=1}^N \int_{S_j}
\alpha_j(y)\,
\bigl(u\upharpoonright_{S_j}\bigr)(y)\,
\bigl(v\upharpoonright_{S_j}\bigr)(y)\,
d\sigma_j(y).
\end{equation}
Its form domain is $D[h_N]=H^1(\re^3)$.
Since $\alpha_j\in L^\infty(S_j)$ and the trace map $H^1(\re^3)\to L^2(S_j)$ is bounded,
each surface term is bounded on $H^1(\re^3)$.
Hence $h_N$ is closed and lower semibounded and defines a unique self--adjoint operator
(which we still denote by $H_N$).

To simplify notation, we pull everything back to $S^2$.
For $u\in H^1(\re^3)$ we define the trace maps $\tau_j:H^1(\re^3)\to L^2(S^2)$ by
$$
(\tau_j u)(\omega):=u(R_j\omega),
\qquad
\omega\in S^2.
$$
We also define the pulled--back coefficient
$$
\widetilde\alpha_j(\omega):=\alpha_j(R_j\omega),
\qquad
\omega\in S^2.
$$
Since $d\sigma_j(R_j\omega)=R_j^2\,d\omega$, the surface contribution can be written as
$$
\int_{S_j}\alpha_j(y)\,u(y)\,v(y)\,d\sigma_j(y)
=
R_j^2\int_{S^2}\widetilde\alpha_j(\omega)\,(\tau_j u)(\omega)\,(\tau_j v)(\omega)\,d\omega.
$$

We set
\begin{eqnarray*}
\Omega_1 &=& \bigl\{|x|<R_1\bigr\},\\
\Omega_k &=& \bigl\{R_{k-1}<|x|<R_k\bigr\},\qquad k=2,\dots,N,\\
\Omega_{N+1} &=& \bigl\{|x|>R_N\bigr\}.
\end{eqnarray*}

Moreover, the operator domain $D(H_N)$ is characterized by the usual interface rules:
$u$ is continuous across each $S_j$, and the radial derivative satisfies
\begin{equation}\label{eq:jump-general}
\partial_r u(R_j+0,\omega)-\partial_r u(R_j-0,\omega)
=
\widetilde\alpha_j(\omega)\,u(R_j,\omega),
\qquad
\omega\in S^2,\quad j=1,\dots,N.
\end{equation}
More precisely,
\begin{eqnarray}\label{eq:dom-HN}
D(H_N)
&=&
\Bigl\{
u\in L^2(\re^3)\ \Bigm|\ 
u\upharpoonright_{\Omega_k}\in H^2(\Omega_k)\ \text{for }k=1,\dots,N+1,\nonumber\\
&&\quad
u\ \text{is continuous across each }S_j,\
\text{and satisfies \eqref{eq:jump-general}}
\Bigr\}.
\end{eqnarray}
On each region $\Omega_k$ the operator acts as $(H_Nu)\upharpoonright_{\Omega_k}=-\Delta(u|_{\Omega_k})$.

We begin with a simple observation excluding positive eigenvalues.
\begin{thm}\label{thm:no-embedded}
Let $H_N$ be the Schr\"odinger operator on $\mathbb R^3$ with finitely many
concentric spherical $\delta$--shell interactions \eqref{eq:HN-formal}, where
$\alpha_j\in L^\infty(S_j;\mathbb R)$.
Then $H_N$ has no eigenvalues in $(0,\infty)$.
\end{thm}
\begin{proof}
Fix $E=k^2>0$ and assume that $H_Nu=Eu$ for some $u\in L^2(\mathbb R^3)$.
Then $u\in D(H_N)$.
In particular, $u\in H^2(\Omega_k)$ on each region $\Omega_k$ away from the shells,
and the traces $u|_{S_j}$ and $\partial_r^\pm u\upharpoonright_{S_j}$ are well-defined.

For almost every $r>0$, the function $u(r,\cdot)$ belongs to $L^2(S^2)$.
We therefore expand
\begin{equation}\label{eq:partial-wave-f}
u(r,\omega)=\sum_{\ell=0}^\infty\sum_{m=-\ell}^{\ell}
\frac{f_{\ell m}(r)}{r}\,Y_{\ell m}(\omega),
\qquad
f_{\ell m}(r)
=
r\int_{S^2}u(r,\omega)\overline{Y_{\ell m}(\omega)}\,d\omega,
\end{equation}
where $\{Y_{\ell m}\}$ is an orthonormal basis of $L^2(S^2)$.
See \cite{AbramowitzStegun}.
By Parseval's identity,
\begin{equation}\label{eq:parseval}
r^2\int_{S^2}|u(r,\omega)|^2\,d\omega
=\sum_{\ell,m}|f_{\ell m}(r)|^2
\quad\text{for a.e.\ }r>0,
\end{equation}
and consequently
\begin{equation}\label{eq:L2-outside}
\int_{|x|>R_N}|u(x)|^2\,dx
=\int_{R_N}^\infty\sum_{\ell,m}|f_{\ell m}(r)|^2\,dr.
\end{equation}

Away from the shells $r=R_1,\dots,R_N$ the equation reduces to
$(-\Delta-k^2)u=0$.
Substituting~\eqref{eq:partial-wave-f}
(equivalently, projecting $(-\Delta-k^2)u=0$ onto the spherical harmonics)
shows that each $f_{\ell m}$ satisfies
\begin{equation}\label{eq:radial-f}
-f_{\ell m}''(r)+\frac{\ell(\ell+1)}{r^2}f_{\ell m}(r)=k^2 f_{\ell m}(r)
\end{equation}
on every interval not containing any $R_j$.

For $r>R_N$ the general solution of~\eqref{eq:radial-f} can be written as
\begin{equation}\label{eq:f-hankel}
f_{\ell m}(r)=A_{\ell m}\,r\,h^{(1)}_{\ell}(kr)
+B_{\ell m}\,r\,h^{(2)}_{\ell}(kr),
\end{equation}
where $h^{(1)}_{\ell}$ and $h^{(2)}_{\ell}$ are the spherical Hankel functions.
Using the asymptotics
\[
r\,h^{(1)}_{\ell}(kr)=\frac{e^{ikr}}{ik}+O(r^{-1}),
\qquad
r\,h^{(2)}_{\ell}(kr)=-\frac{e^{-ikr}}{ik}+O(r^{-1}),
\quad r\to\infty,
\]
we obtain
\begin{equation}\label{eq:f-asymp}
f_{\ell m}(r)=p_{\ell m}e^{ikr}
+q_{\ell m}e^{-ikr}+O(r^{-1}),
\qquad r\to\infty,
\end{equation}
for suitable constants $p_{\ell m},q_{\ell m}$.

Assume that $u$ does not vanish identically in the exterior region $(R_N,\infty)$.
Then there exists $(\ell_0,m_0)$ such that $f_{\ell_0m_0}$ is not identically zero
on $(R_N,\infty)$.
Write $f=f_{\ell_0m_0}$ and $g(r)=p e^{ikr}+q e^{-ikr}$ with
$M:=|p|^2+|q|^2>0$.
By~\eqref{eq:f-asymp} there exist $C>0$ and $R_0>R_N$ such that
$|f(r)-g(r)|\le C/r$ for $r\ge R_0$.

For every $R\ge R_0$ we have
\[
\int_R^{R+\pi/k}|g(r)|^2\,dr
=\int_R^{R+\pi/k}\Bigl(M+2\Re(p\overline{q}e^{2ikr})\Bigr)\,dr
=\frac{\pi}{k}M,
\]
since $\int_R^{R+\pi/k}e^{2ikr}\,dr=0$.
Moreover, using $|a+b|^2\ge \frac12|a|^2-|b|^2$, we obtain
\begin{eqnarray*}
\int_R^{R+\pi/k}|f(r)|^2\,dr
&\ge&
\frac12\int_R^{R+\pi/k}|g(r)|^2\,dr
-\int_R^{R+\pi/k}|f(r)-g(r)|^2\,dr \\
&\ge&
\frac{\pi}{2k}M-\frac{C^2\pi}{kR^2}.
\end{eqnarray*}
Choosing $R$ so large that $C^2/R^2\le M/4$, we obtain
\[
\int_R^{R+\pi/k}|f(r)|^2\,dr\ge \frac{\pi}{4k}M=:c>0.
\]
Summing over the disjoint intervals
$[R+n\pi/k,R+(n+1)\pi/k]$ yields
\[
\int_R^\infty |f(r)|^2\,dr=\infty.
\]
Using~\eqref{eq:L2-outside} and
$\sum_{\ell,m}|f_{\ell m}(r)|^2\ge |f(r)|^2$,
we conclude that
\[
\int_{|x|>R_N}|u(x)|^2\,dx=\infty,
\]
contradicting $u\in L^2(\mathbb R^3)$.
Therefore $u\equiv0$ for $r>R_N$.

By continuity, $u|_{S_N}=0$, and since $u$ vanishes identically in the exterior
region we also have $\partial_r^+u|_{S_N}=0$.
The jump condition \eqref{eq:jump-general} on $S_N$ implies
$\partial_r^-u|_{S_N}=0$.
Hence $f_{\ell m}(R_N)=0$ for all $\ell,m$.
Moreover, since $u\in H^2(\Omega_N)\cap H^2(\Omega_{N+1})$, we have
\[
f'_{\ell m}(R_N\pm 0)
=R_N\int_{S^2}\partial_r^\pm u(R_N,\omega)
\overline{Y_{\ell m}(\omega)}\,d\omega,
\]
and thus $f'_{\ell m}(R_N)=0$ for all $\ell,m$.
Uniqueness for the second--order ODE~\eqref{eq:radial-f} yields
$f_{\ell m}\equiv0$ on $(R_{N-1},R_N)$.
Repeating the argument across the remaining shells shows that
$u\equiv0$ on $\mathbb R^3$, a contradiction.
Thus there are no positive eigenvalues.
\end{proof}

When the surface strengths are constants, rotational symmetry allows a
partial--wave decomposition.
At the zero--energy threshold $E=0$, each angular momentum channel reduces to a
finite--dimensional algebraic condition on the shell values, which can be
written explicitly.

\begin{thm}[Zero--energy threshold]\label{thm:zero-energy}
Fix an integer $N\ge1$ and radii $0<R_1<\cdots<R_N$.
Assume that the surface strengths are constants $\alpha_1,\dots,\alpha_N\in\mathbb R$.
Let $\ell\in\{0,1,2,\dots\}$ be the angular--momentum index.

{\rm (i)} In the $s$--wave case $\ell=0$, the point $E=0$ is not an $L^2$
eigenvalue for any finite $\alpha_1,\dots,\alpha_N$.

{\rm (ii)} For each fixed $\ell\ge1$, define the $N\times N$ matrix
$A_\ell=(a^{(\ell)}_{ij})_{i,j=1}^N$ by
\begin{equation}\label{eq:Aell-entry}
a^{(\ell)}_{ij}
:=
\delta_{ij}
+
\frac{\alpha_j}{2\ell+1}\,
\frac{R_{\min(i,j)}^{\ell+1}}{R_{\max(i,j)}^{\ell}},
\qquad i,j=1,\dots,N,
\end{equation}
where $R_{\min(i,j)}:=\min(R_i,R_j)$ and $R_{\max(i,j)}:=\max(R_i,R_j)$.
Then $E=0$ is an $L^2$ eigenvalue in the $\ell$--th partial wave if and only if
\begin{equation}\label{eq:zero-cond}
\det A_\ell=0.
\end{equation}
Moreover, the multiplicity of $E=0$ contributed by the $\ell$--th partial wave
equals $(2\ell+1)\dim\ker A_\ell$.
\end{thm}

\begin{proof}
Fix $\ell\ge0$ and write
$$
\psi(x)=\frac{u(r)}{r}Y_{\ell m}(\omega),\qquad r=|x|,\ \omega=x/r,
$$
so that on each region the radial function satisfies
$$
-\,u''(r)+\frac{\ell(\ell+1)}{r^2}u(r)=E\,u(r),
\qquad r\ne R_1,\dots,R_N.
$$
Moreover, since $\{Y_{\ell m}\}$ is an orthonormal basis of $L^2(S^2)$, we have
$\|\psi\|_{L^2(\mathbb R^3)}^2=\int_0^\infty |u(r)|^2\,dr$.
The $\delta$--shell interface conditions for $\psi$ are equivalent to the
continuity of $u$ at each $R_j$ and to the jump conditions
$$
u'(R_j+0)-u'(R_j-0)=\alpha_j\,u(R_j),\qquad j=1,\dots,N.
$$

We consider $E=0$.

{\rm (i)} Let $\ell=0$.
Then $-u''(r)=0$ on each interval, so on the exterior region $r>R_N$ we have
$u(r)=ar+b$.
Since $\int_{R_N}^\infty |u(r)|^2\,dr<\infty$, it follows that $a=b=0$ and hence
$u\equiv0$ on $(R_N,\infty)$.
By continuity, $u(R_N)=0$, and the jump condition at $R_N$ yields
$u'(R_N-0)=u'(R_N+0)=0$.
Since $-u''(r)=0$ on $(R_{N-1},R_N)$, these two boundary conditions imply that
$u\equiv0$ on $(R_{N-1},R_N)$.
Iterating this argument across the remaining shells, we obtain $u\equiv0$ on
$(0,\infty)$, hence $\psi\equiv0$, a contradiction.
Therefore $E=0$ cannot be an $L^2$ eigenvalue in the $s$--wave sector.

{\rm (ii)} Assume $\ell\ge1$.
Consider the homogeneous equation
$$
-\,w''(r)+\frac{\ell(\ell+1)}{r^2}w(r)=0\qquad (r\ne s),
$$
whose independent solutions are $r^{\ell+1}$ and $r^{-\ell}$.
Let $r_<:=\min(r,s)$ and $r_>:=\max(r,s)$ and define
\begin{equation}\label{eq:Green0}
G_\ell(r,s):=\frac{1}{2\ell+1}\,r_<^{\ell+1}r_>^{-\ell}.
\end{equation}
Then $G_\ell(\cdot,s)$ is continuous, solves the homogeneous equation for $r\ne s$,
decays like $r^{-\ell}$ as $r\to\infty$, is regular at $r=0$, and satisfies
$G'_\ell(s+0,s)-G'_\ell(s-0,s)=-1$.
Consequently, for any constants $c_1,\dots,c_N$, the function
$$
w(r):=-\sum_{j=1}^N c_j\,G_\ell(r,R_j)
$$
is continuous on $(0,\infty)$, solves the homogeneous equation away from the
shells, and has derivative jumps
$w'(R_j+0)-w'(R_j-0)=c_j$.

Let $\psi$ be an $L^2$ solution at $E=0$ in the $\ell$--th partial wave and set
$U_j:=u(R_j)$.
Define
$$
\widetilde u(r):=-\sum_{j=1}^N \alpha_j U_j\,G_\ell(r,R_j).
$$
By the preceding properties of $G_\ell$, the function $\widetilde u$ is continuous,
solves the radial equation for $r\ne R_j$, and satisfies
$\widetilde u'(R_j+0)-\widetilde u'(R_j-0)=\alpha_j U_j$.
Hence $v:=u-\widetilde u$ solves the homogeneous equation on $(0,\infty)$,
is continuous with no derivative jumps at any $R_j$, is regular at $0$,
and belongs to $L^2(0,\infty)$.
Therefore $v\equiv0$, and thus $u=\widetilde u$.
Evaluating at $r=R_i$ gives, for $i=1,\dots,N$,
$$
U_i
=
-\sum_{j=1}^N \alpha_j U_j\,G_\ell(R_i,R_j)
=
-\sum_{j=1}^N \alpha_j U_j\,
\frac{1}{2\ell+1}\,
\frac{R_{\min(i,j)}^{\ell+1}}{R_{\max(i,j)}^{\ell}}.
$$
This is exactly the linear system $A_\ell U=0$ with $A_\ell$ as in
\eqref{eq:Aell-entry}. Hence a nontrivial $L^2$ solution exists if and only if
$\ker A_\ell\ne\{0\}$, equivalently \eqref{eq:zero-cond} holds.

Finally, the radial equation and the interface conditions do not depend on $m$,
so each independent radial solution produces $(2\ell+1)$ linearly independent
eigenfunctions by varying $m$. This gives the multiplicity formula
$(2\ell+1)\dim\ker A_\ell$.
\end{proof}

\begin{rem}
For $N=2$ with $0<R_1<R_2$ and constants $\alpha_1,\alpha_2\in\mathbb R$, the
matrix $A_\ell$ in \eqref{eq:Aell-entry} takes the form
$$
A_\ell=
\left(\begin{array}{cc}
1+\dfrac{\alpha_1 R_1}{2\ell+1}
&
\dfrac{\alpha_2}{2\ell+1}\dfrac{R_1^{\ell+1}}{R_2^{\ell}}
\\[3mm]
\dfrac{\alpha_1}{2\ell+1}\dfrac{R_1^{\ell+1}}{R_2^{\ell}}
&
1+\dfrac{\alpha_2 R_2}{2\ell+1}
\end{array}\right).
$$
Thus the condition $\det A_\ell=0$ becomes
$$
\alpha_1\alpha_2\,R_1^{2\ell+2}R_2^{-2\ell}
=
\bigl(\alpha_1 R_1 + 2\ell+1\bigr)\bigl(\alpha_2 R_2 + 2\ell+1\bigr).
$$
\end{rem}

\section{Resolvent Formula}\label{sec:elem-resolvent}

In this section we derive an explicit resolvent formula for the $N$--shell
Hamiltonian by using the free Green kernel and single--layer potentials.
We do not use boundary triples or more abstract extension theory.
Throughout, we take the principal branch so that $\Im\sqrt z>0$, and we set
$k=\sqrt z$.

\medskip
Fix $z\in\co\setminus[0,\infty)$ and write
$$
G_z(x,y)=\frac{e^{ik|x-y|}}{4\pi|x-y|},\qquad
(R_0(z)f)(x)=\int_{\re^3}G_z(x,y)f(y)\,dy .
$$
For $u\in H^1(\re^3)$ we define the pulled--back trace maps
$\tau_j:H^1(\re^3)\to L^2(S^2)$ by
$$
(\tau_j u)(\omega)=u(R_j\omega),\qquad \omega\in S^2.
$$
Under the parametrization $y=R_j\omega$, $\omega\in S^2$, one has
$d\sigma_j(y)=R_j^2\,d\omega$, hence
$$
\int_{S_j}|u(y)|^2\,d\sigma_j(y)
=
R_j^2\int_{S^2}|(\tau_j u)(\omega)|^2\,d\omega.
$$
We also write
$$
\widetilde\alpha_j(\omega)=\alpha_j(R_j\omega),\qquad \omega\in S^2.
$$

For $\varphi\in L^2(S^2)$ we define the single--layer potentials supported on
$S_j$ by
\begin{equation}\label{eq:Gammaj}
(\Gamma_j(z)\varphi)(x)
=
\int_{S^2}G_z\bigl(x,R_j\omega\bigr)\,\varphi(\omega)\,d\omega,
\qquad x\in\re^3.
\end{equation}

We introduce the boundary Hilbert space
\begin{equation}\label{eq:HN-boundary-space}
\mathcal K_N := \bigoplus_{j=1}^N L^2(S^2),
\end{equation}
whose elements are written as ${}^{t}(\varphi_1,\dots,\varphi_N)$ with
$\varphi_j\in L^2(S^2)$.
On $\mathcal K_N$ we consider the $N\times N$ operator matrix
$m(z)=(m_{ij}(z))_{i,j=1}^N$ defined by
\begin{equation}\label{eq:Mij}
m_{ij}(z):L^2(S^2)\to L^2(S^2),\qquad
(m_{ij}(z)\varphi)(\omega)
=
\int_{S^2}G_z\bigl(R_i\omega,R_j\omega'\bigr)\,\varphi(\omega')\,d\omega'.
\end{equation}
We also write $\Gamma(z)=(\Gamma_1(z),\dots,\Gamma_N(z))$ as a column operator
$$
\Gamma(z):\mathcal K_N\to L^2(\re^3),\qquad
\Gamma(z)\,{}^{t}(\varphi_1,\dots,\varphi_N)
=
\sum_{j=1}^N \Gamma_j(z)\varphi_j .
$$
Finally we define a bounded diagonal operator $\Theta$ on $\mathcal K_N$ by
\begin{equation}\label{eq:ThetaN}
(\Theta\,{}^{t}(\varphi_1,\dots,\varphi_N))_j(\omega)
=
R_j^2\,\widetilde\alpha_j(\omega)\,\varphi_j(\omega),
\qquad j=1,\dots,N,
\end{equation}
and set
\begin{equation}\label{eq:K-def}
K_N(z)=I+m(z)\Theta .
\end{equation}

\begin{lem}\label{lem:gamma-HS}
Let $R>0$ and $z\in\co\setminus[0,\infty)$ with $\Im\sqrt z>0$.
Define
$$
(\widetilde{\Gamma}(z)\varphi)(x)
=
\int_{S^2}G_z(x,R\omega)\,\varphi(\omega)\,d\omega,
\qquad \varphi\in L^2(S^2).
$$
Then $\widetilde{\Gamma}(z):L^2(S^2)\to L^2(\re^3)$ is a Hilbert--Schmidt operator.
\end{lem}

\begin{proof}
The Hilbert--Schmidt norm is
$$
\|\widetilde{\Gamma}(z)\|_{\mathrm{HS}}^2
=
\int_{\re^3}\int_{S^2}\bigl|G_z(x,R\omega)\bigr|^2\,d\omega\,dx .
$$
Let $c=\Im\sqrt z>0$.
Since
$$
\bigl|G_z(x,y)\bigr|
\le
\frac{1}{4\pi}\,\frac{e^{-c|x-y|}}{|x-y|}
\qquad(x\ne y),
$$
we obtain
$$
\|\widetilde{\Gamma}(z)\|_{\mathrm{HS}}^2
\le
\frac{1}{(4\pi)^2}\int_{\re^3}\int_{S^2}
\frac{e^{-2c|x-R\omega|}}{|x-R\omega|^2}\,d\omega\,dx .
$$
Write $r=|x|$ and let $\theta$ be the angle between $x/|x|$ and $\omega$.
Then $\rho=|x-R\omega|=\sqrt{r^2+R^2-2rR\cos\theta}$, and
$d\omega=2\pi\sin\theta\,d\theta$ with $\sin\theta\,d\theta=(\rho/(rR))\,d\rho$.
Hence
$$
\int_{S^2}\frac{e^{-2c|x-R\omega|}}{|x-R\omega|^2}\,d\omega
=
\frac{2\pi}{rR}\int_{|r-R|}^{r+R}\frac{e^{-2c\rho}}{\rho}\,d\rho .
$$
Using $dx=4\pi r^2\,dr$ and Fubini--Tonelli, we get
\begin{eqnarray*}
\|\widetilde{\Gamma}(z)\|_{\mathrm{HS}}^2
&\le&
\frac{1}{(4\pi)^2}
\int_0^\infty\left(
\frac{2\pi}{rR}
\int_{|r-R|}^{r+R}\frac{e^{-2c\rho}}{\rho}\,d\rho\right)
4\pi r^2\,dr \\
&=&
\frac{1}{2R}\int_0^\infty r
\int_{|r-R|}^{r+R}\frac{e^{-2c\rho}}{\rho}\,d\rho\,dr.
\end{eqnarray*}
The domain is
$$
D=\Bigl\{(r,\rho):\ r\ge0,\ |r-R|\le \rho \le r+R\Bigr\}
=
\Bigl\{(r,\rho):\ \rho\ge0,\ |R-\rho|\le r \le R+\rho\Bigr\}.
$$
Exchanging the order of integration yields
\begin{eqnarray*}
\|\widetilde{\Gamma}(z)\|_{\mathrm{HS}}^2
&\le&
\frac{1}{2R}\int_0^\infty \frac{e^{-2c\rho}}{\rho}
\left(\int_{|R-\rho|}^{R+\rho} r\,dr\right)d\rho
=
\int_0^\infty e^{-2c\rho}\,d\rho
=
\frac{1}{2c}
<
\infty .
\end{eqnarray*}
\end{proof}

\begin{lem}\label{lem:maps}
Let $z\in\co\setminus[0,\infty)$ and $\Im\sqrt z>0$.
Then the following statements hold.

{\rm (a)} Each $\Gamma_j(z):L^2(S^2)\to L^2(\re^3)$ is a Hilbert--Schmidt
operator, and in particular $\Gamma(z)$ is Hilbert--Schmidt.

{\rm (b)} For every $i,j$ the operator $m_{ij}(z)$ is bounded on $L^2(S^2)$ and
one has $m_{ij}(\bar z)=m_{ji}(z)^{\ast}$.
Moreover, the identity $\tau_i\Gamma_j(z)=m_{ij}(z)$ holds in $L^2(S^2)$.

{\rm (c)} The adjoint satisfies $\Gamma_j(\bar z)^{\ast}=\tau_j R_0(z)$
as an operator from $L^2(\re^3)$ to $L^2(S^2)$.
\end{lem}

\begin{proof}
(a) follows from Lemma~\ref{lem:gamma-HS} with $R=R_j$.

(b) The operator $m_{ij}(z)$ is an integral operator with kernel
$K_{ij}^z(\omega,\omega')=G_z(R_i\omega,R_j\omega')$.
Let $c=\Im\sqrt z>0$.
Then
$$
\bigl|K_{ij}^z(\omega,\omega')\bigr|
\le
\frac{1}{4\pi}\,\frac{e^{-c|R_i\omega-R_j\omega'|}}{|R_i\omega-R_j\omega'|}
\qquad(\omega\ne\omega').
$$
For fixed $\omega$, the map $\omega'\mapsto |R_i\omega-R_j\omega'|^{-1}$ is
integrable on $S^2$, and similarly with $\omega$ and $\omega'$ interchanged.
Hence the Schur test yields boundedness of $m_{ij}(z)$ on $L^2(S^2)$.
The symmetry $G_{\bar z}(x,y)=\overline{G_z(y,x)}$ implies
$m_{ij}(\bar z)=m_{ji}(z)^{\ast}$.

To prove $\tau_i\Gamma_j(z)=m_{ij}(z)$, let $\phi\in L^2(S^2)$.
The estimate above implies absolute integrability on $S^2\times S^2$, hence
Fubini--Tonelli applies and, for a.e.\ $\omega\in S^2$,
$$
(\tau_i\Gamma_j(z)\phi)(\omega)
=
\int_{S^2}G_z(R_i\omega,R_j\omega')\,\phi(\omega')\,d\omega'
=
(m_{ij}(z)\phi)(\omega).
$$

(c) Let $f\in L^2(\re^3)$ and $\phi\in L^2(S^2)$.
Using Fubini--Tonelli (justified by Cauchy--Schwarz and (a)), we compute
\begin{eqnarray*}
(\Gamma_j(\bar z)\phi,f)_{L^2(\re^3)}
&=&\int_{\re^3}\int_{S^2}
   G_{\bar z}(x,R_j\omega)\,\phi(\omega)\,
   \overline{f(x)}\,d\omega\,dx \\
&=&\int_{S^2}\phi(\omega)\,
   \overline{\int_{\re^3}
   \overline{G_{\bar z}(x,R_j\omega)}\,f(x)\,dx}\,d\omega \\
&=&\int_{S^2}\phi(\omega)\,
   \overline{\int_{\re^3}
   G_{z}(R_j\omega,x)\,f(x)\,dx}\,d\omega
=
(\phi,\ \tau_j R_0(z)f)_{L^2(S^2)} .
\end{eqnarray*}
Hence $\Gamma_j(\bar z)^{\ast}=\tau_j R_0(z)$.
\end{proof}

We recall the quadratic form $h_N$ from Section~\ref{sec:form},
\begin{equation}\label{eq:hN-recall}
h_N[u,v]
=
\int_{\re^3}\nabla u\cdot\overline{\nabla v}\,dx
+
\sum_{j=1}^{N}\int_{S_j}\alpha_j(x)\,u(x)\,\overline{v(x)}\,d\sigma_j(x),
\qquad
D[h_N]=H^1(\re^3),
\end{equation}
which defines a self--adjoint operator $H_N$ in $L^2(\re^3)$.

While Kre\u{\i}n--type resolvent formulas are available in abstract extension
theory, we present here a self--contained boundary integral derivation tailored
to concentric spheres, which yields an explicit $N\times N$ operator matrix
$K_N(z)$ on $\mathcal K_N=\bigoplus_{j=1}^N L^2(S^2)$ and is directly compatible with the
subsequent partial--wave reduction.

We are now in a position to state the resolvent formula for $H_N$, 
which is one of the main results of this paper.
\begin{thm}\label{thm:elem-res}
Let $H_N$ be as above and let $z\in\co\setminus[0,\infty)$ with $\Im\sqrt z>0$.
Assume that $K_N(z)=I+m(z)\Theta$ on $\bigoplus_{j=1}^N L^2(S^2)$ is invertible,
where $\Theta$ is given in \eqref{eq:ThetaN}.
Define a bounded operator $R(z)$ on $L^2(\re^3)$ by
\begin{equation}\label{eq:elem-res}
R(z)
=
R_0(z)
-
\Gamma(z)\,\Theta\,K_N(z)^{-1}\Gamma(\bar z)^{\ast}.
\end{equation}
Then $R(z)=(H_N-z)^{-1}$, in particular $z\in\rho(H_N)$, and the resolvent difference
$(H_N-z)^{-1}-R_0(z)$ is a trace class operator.
\end{thm}

\begin{proof}
Assume that $K_N(z)$ is invertible and define $R(z)$ by \eqref{eq:elem-res}.
Let $f\in L^2(\re^3)$ and set
\begin{equation}\label{eq:def-u-Rzf}
u:=R(z)f.
\end{equation}
We show that $u\in D(H_N)$ and $(H_N-z)u=f$.

We first verify that $u\in H^1(\re^3)$.
Since $R_0(z):L^2(\re^3)\to H^2(\re^3)\subset H^1(\re^3)$, we have
$R_0(z)f\in H^1(\re^3)$.
Each trace map $\tau_j:H^1(\mathbb R^3)\to H^{1/2}(S^2)$ is bounded, and we use the continuous embedding
$H^{1/2}(S^2)\hookrightarrow L^2(S^2)$ to regard $\tau_j$ as a bounded operator into $L^2(S^2)$.
Hence $\tau_j:H^1(\re^3)\to L^2(S^2)$ is bounded, and its adjoint
$\tau_j^\ast:L^2(S^2)\to H^{-1}(\re^3)$ is bounded.
For $z\in\co\setminus[0,\infty)$ the free resolvent extends boundedly as
$R_0(z):H^{-1}(\re^3)\to H^1(\re^3)$.
By Lemma~\ref{lem:maps}, one has the identity $R_0(z)\tau_j^\ast=\Gamma_j(z)$.
Since $R_0(z):H^{-1}(\mathbb R^3)\to H^1(\mathbb R^3)$ is bounded, it follows that
$\Gamma_j(z)$ maps $L^2(S^2)$ boundedly into $H^1(\mathbb R^3)$.
Hence \eqref{eq:def-u-Rzf} implies $u\in H^1(\mathbb R^3)$.
Set
\[
g:=K_N(z)^{-1}\Gamma(\bar z)^{\ast}f
\in \bigoplus_{j=1}^N L^2(S^2),
\qquad
g={}^{t}(g_1,\dots,g_N).
\]
Using again Lemma~\ref{lem:maps}, namely $\Gamma_j(\bar z)^{\ast}=\tau_j R_0(z)$, we obtain
\[
u
=
R_0(z)f-\sum_{j=1}^N \Gamma_j(z)\,(\Theta g)_j,
\qquad
(\Theta g)_j(\omega)=R_j^2\,\widetilde\alpha_j(\omega)\,g_j(\omega).
\]
Taking traces and using $\tau_i\Gamma_j(z)=m_{ij}(z)$ from Lemma~\ref{lem:maps}(b), we obtain
$$
\tau_i u
=
\tau_iR_0(z)f-\sum_{j=1}^N m_{ij}(z)\,(\Theta g)_j,
\qquad i=1,\dots,N.
$$
In vector form, with $\tau u={}^{t}(\tau_1u,\dots,\tau_Nu)$, this reads
$$
\tau u
=
\Gamma(\bar z)^{\ast}f-m(z)\Theta g.
$$
Since $K_N(z)=I+m(z)\Theta$ and $K_N(z)g=\Gamma(\bar z)^{\ast}f$, we obtain
$$
\tau u
=
K_N(z)g-m(z)\Theta g
=
g.
$$

Let $\varphi\in H^1(\re^3)$ and let $H_0:=-\Delta$ be the free Schr\"odinger operator. 
Since $(H_0-z)R_0(z)f=f$, we have
$$
(\nabla R_0(z)f,\nabla\varphi)_{L^2(\re^3)}
-
z(R_0(z)f,\varphi)_{L^2(\re^3)}
=
(f,\varphi)_{L^2(\re^3)}.
$$
Moreover, for $h\in L^2(S^2)$ one has $(H_0-z)\Gamma_j(z)h=\tau_j^\ast h$ in
$H^{-1}(\re^3)$, hence
$$
(\nabla \Gamma_j(z)h,\nabla\varphi)_{L^2(\re^3)}
-
z(\Gamma_j(z)h,\varphi)_{L^2(\re^3)}
=
(h,\tau_j\varphi)_{L^2(S^2)} .
$$
Applying this with $h=(\Theta g)_j$ and summing over $j$, we obtain
$$
(\nabla u,\nabla\varphi)_{L^2(\re^3)}
-
z(u,\varphi)_{L^2(\re^3)}
=
(f,\varphi)_{L^2(\re^3)}
-
\sum_{j=1}^N ((\Theta g)_j,\tau_j\varphi)_{L^2(S^2)}.
$$
On the other hand, by \eqref{eq:hN-recall} and $d\sigma_j=R_j^2\,d\omega$,
$$
h_N[u,\varphi]-z(u,\varphi)_{L^2(\re^3)}
=
(\nabla u,\nabla\varphi)_{L^2(\re^3)}
-
z(u,\varphi)_{L^2(\re^3)}
+
\sum_{j=1}^N ((\Theta\,\tau u)_j,\tau_j\varphi)_{L^2(S^2)}.
$$
Since $\tau u=g$, we have $\Theta\,\tau u=\Theta g$, and hence
$$
h_N[u,\varphi]-z(u,\varphi)_{L^2(\re^3)}
=
(f,\varphi)_{L^2(\re^3)}
\qquad(\forall\,\varphi\in H^1(\re^3)).
$$
Since $h_N$ is closed and lower semibounded, by the representation theorem
(with $w:=f+zu\in L^2(\re^3)$) this implies that
$u\in D(H_N)$ and $(H_N-z)u=f$.
Therefore $R(z)=(H_N-z)^{-1}$.

Finally, by Lemma~\ref{lem:maps}(a), both $\Gamma(z)$ and $\Gamma(\bar z)^{\ast}$ are
Hilbert--Schmidt operators.
Since $\Theta$ and $K_N(z)^{-1}$ are bounded, the composition
$\Gamma(z)\Theta K_N(z)^{-1}\Gamma(\bar z)^{\ast}$ is a trace class operator,
by the ideal property of the Schatten--von Neumann classes.
Consequently, the resolvent difference $(H_N-z)^{-1}-R_0(z)$ is trace class.
\end{proof}

\begin{lem}\label{lem:compact-m}
For each $z\in\co\setminus[0,\infty)$ the operator $m(z)$ acts boundedly on
$\bigoplus_{j=1}^N L^2(S^2)$, depends analytically on $z$, and is compact.
Hence $K_N(z)=I+m(z)\Theta$ is an analytic Fredholm family of index zero.
In particular, if $K_N(z_0)$ is invertible for some $z_0\in\co\setminus[0,\infty)$,
then $K_N(z)^{-1}$ exists and depends meromorphically on $z$, and the
set of points where $K_N(z)$ is not invertible is discrete.
These points coincide with the poles of the resolvent of $H_N$; see
Theorem~\ref{thm:elem-res} together with the resolvent identity
\eqref{eq:elem-res}.
\end{lem}

\begin{proof}
Each block $m_{ij}(z)$ is an integral operator with kernel
$G_z(R_i\omega,R_j\omega')$ on the compact set $S^2\times S^2$.

If $i\ne j$, then $|R_i\omega-R_j\omega'|\ge|R_i-R_j|>0$, so the kernel is continuous
on $S^2\times S^2$ for $\Im\sqrt z>0$.
Hence $m_{ij}(z)$ is Hilbert--Schmidt and in particular compact.

For $i=j$, we write
$$
m_{ii}(z)=m_{ii}(0)+\bigl(m_{ii}(z)-m_{ii}(0)\bigr).
$$
The difference $G_z-G_0$ has a removable singularity at the diagonal, since
$$
\frac{e^{i\sqrt z r}-1}{4\pi r}
=
\frac{i\sqrt z}{4\pi}+O(r)
\qquad(r\to0).
$$
Thus the kernel of $m_{ii}(z)-m_{ii}(0)$ extends continuously to $S^2\times S^2$,
so $m_{ii}(z)-m_{ii}(0)$ is Hilbert--Schmidt and hence compact.
On the other hand, $m_{ii}(0)$ is the Laplace single--layer operator on $S^2$.
It maps $L^2(S^2)$ continuously into $H^1(S^2)$, and the embedding
$H^1(S^2)\hookrightarrow L^2(S^2)$ is compact.
Therefore $m_{ii}(0)$ is compact on $L^2(S^2)$, and so is $m_{ii}(z)$.

Hence all blocks $m_{ij}(z)$ are compact and $m(z)$ is compact on
$\mathcal K_N$.
Analyticity in $z$ follows from the explicit kernel dependence for $\Im\sqrt z>0$.
Since $\Theta$ is bounded, $m(z)\Theta$ is compact and $K_N(z)=I+m(z)\Theta$ is an
analytic Fredholm family of index zero.
The remaining statements follow from the analytic Fredholm theorem (applied once
$K_N(z_0)$ is invertible for some $z_0\in\co\setminus[0,\infty)$) together with
Theorem~\ref{thm:elem-res} and \eqref{eq:elem-res}.
\end{proof}

The following proposition characterizes the eigenvalues of $H_N$ in terms of
the boundary operator $K_N(z)$.

\begin{prop}\label{prop:eig-cond}
Let $z\in\co\setminus[0,\infty)$.
Then $z$ belongs to the point spectrum of $H_N$ if and only if $K_N(z)$ is not
invertible.
Moreover, if $K_N(z)$ is not invertible, then
\[
\dim\ker(H_N-z)=\dim\ker K_N(z).
\]
\end{prop}

\begin{proof}
If $K_N(z)$ is invertible, then Theorem~\ref{thm:elem-res} implies
$z\in\rho(H_N)$, hence $z$ is not an eigenvalue.

Assume that $K_N(z)$ is not invertible and take $0\ne g\in\ker K_N(z)$.
Set
$$
u:=-\Gamma(z)\Theta g.
$$
As in the proof of Theorem~\ref{thm:elem-res}, we have $u\in H^1(\re^3)$.
Moreover, by Lemma~\ref{lem:maps}(b),
$$
\tau u=-m(z)\Theta g,
$$
and $K_N(z)g=0$ is equivalent to $g=-m(z)\Theta g$, hence $\tau u=g$.
Repeating the form computation in the proof of Theorem~\ref{thm:elem-res} with $f=0$
gives
$$
h_N[u,\varphi]-z(u,\varphi)_{L^2(\re^3)}=0
\qquad(\forall\,\varphi\in H^1(\re^3)).
$$
Thus $u\in D(H_N)$ and $(H_N-z)u=0$, so $z$ is an eigenvalue.

Conversely, let $0\ne u\in\ker(H_N-z)$ and set $g:=\tau u$.
Arguing as in the proof of Theorem~\ref{thm:elem-res} with $f=0$ (i.e., using the
identity \eqref{eq:elem-res} for the resolvent difference) shows that
$u=-\Gamma(z)\Theta g$ and $K_N(z)g=0$, hence $K_N(z)$ is not invertible; cf.
Lemma~\ref{lem:maps} and the form characterization \eqref{eq:hN-recall}.

Finally, the maps $g\mapsto u=-\Gamma(z)\Theta g$ and $u\mapsto\tau u$ restrict to
inverse bijections between $\ker K_N(z)$ and $\ker(H_N-z)$, and therefore
$\dim\ker(H_N-z)=\dim\ker K_N(z)$.
\end{proof}

Before turning to rotational symmetry, we note that trace class perturbations of
the resolvent are closely related to scattering theory.
In the present setting Theorem~\ref{thm:elem-res} shows that
$(H_N-z)^{-1}-R_0(z)$ is of trace class for all $z\in\rho(H_N)$ with
$z\in\co\setminus[0,\infty)$.
By the Birman--Kuroda theorem~\cite{ReedSimonIII}, this implies existence and
completeness of the wave operators and unitary equivalence of the absolutely
continuous parts.
As a consequence, the absolutely continuous and essential spectra of $H_N$ can
be identified explicitly, as stated in the following theorem.

\begin{thm}\label{thm:ac}
The absolutely continuous and essential spectra of $H_N$ coincide with those
of the free Hamiltonian:
$$
\sigma_{\mathrm{ac}}(H_N)=[0,\infty),\qquad
\sigma_{\mathrm{ess}}(H_N)=[0,\infty).
$$
Moreover, the negative spectrum of $H_N$ is purely discrete and consists of
eigenvalues of finite multiplicity; the only possible accumulation point of
the negative eigenvalues is $0$.
\end{thm}

\begin{rem}
The operator $K_N(z)$ acts on an infinite--dimensional Hilbert space, so in general
a global determinant is not defined.
The correct spectral condition is the non--invertibility of $K_N(z)$.
If we assume rotational symmetry, namely that each $\alpha_j$ is a constant, then
$\Theta$ acts only on the shell index and the kernels of $m_{ij}(z)$ are rotation
invariant.
In that case one can reduce $K_N(z)$ by the spherical harmonic decomposition and
obtain finite--dimensional secular equations in each angular momentum channel.
\end{rem}

In the remainder of the paper, when we discuss explicit closed formulas and
partial--wave reduction, we assume rotational symmetry, that is, $\alpha_j$ are
constants.

\begin{rem}\label{rem:ml-explicit}
Under rotational symmetry the $\ell$--channel block $m_\ell(z)$ admits an explicit
closed form in terms of spherical Bessel and Hankel functions.
For the reader's convenience we record the formula in Appendix~B, together with
the examples $N=2$ and $\ell=0,1$.
\end{rem}

\begin{lem}\label{lem:block-reduction}
Let $\{Y_{\ell m}:-\ell\le m\le\ell\}$ be a complete orthonormal basis of
spherical harmonics on $S^2$, and set
$$
 \mathcal K_N
 = \bigoplus_{j=1}^N L^2(S^2)
 = \bigoplus_{\ell=0}^\infty \mathcal H_\ell^{\oplus N},
 \qquad
 \mathcal H_\ell = \operatorname{span}\{Y_{\ell m}:-\ell\le m\le\ell\}.
$$
Suppose that the entries of $m(z)$ have rotation--invariant kernels
$k_{ij}(\omega\cdot\omega')$, $1\le i,j\le N$, and that $\Theta$ acts only on the
shell index.
Then for each $\ell\ge0$ there is an $N\times N$ matrix $m_\ell(z)$ such that
$$
 m(z)\big|_{\mathcal H_\ell^{\oplus N}}
 =
 m_\ell(z)\otimes I_{\mathcal H_\ell}.
$$
In particular, $K_N(z)$ fails to be invertible if and only if
$$
 \det\bigl(I_N + m_\ell(z)\Theta\bigr)=0
$$
for some $\ell\ge0$, and the corresponding eigenspace has dimension
$(2\ell+1)\,\dim\ker\bigl(I_N+m_\ell(z)\Theta\bigr)$.
\end{lem}

\begin{proof}
The spherical harmonic addition theorem (see, e.g., \cite{AbramowitzStegun})
states that
$$
 P_\ell(\omega\cdot\omega')
 =
 \frac{4\pi}{2\ell+1}
 \sum_{m=-\ell}^{\ell}
 Y_{\ell m}(\omega)\,\overline{Y_{\ell m}(\omega')},
 \qquad \omega,\omega'\in S^2,
$$
where $P_\ell$ is the Legendre polynomial of degree $\ell$.
If we expand each rotation--invariant kernel $k_{ij}$ in Legendre polynomials
and use the addition theorem, then $m(z)$ is diagonal in the spherical harmonic
basis and does not mix different values of $\ell$ or $m$.
More precisely, $m(z)$ leaves $\mathcal H_\ell^{\oplus N}$ invariant and acts
there as
$$
 m(z)\big|_{\mathcal H_\ell^{\oplus N}}
 =
 m_\ell(z)\otimes I_{\mathcal H_\ell},
$$
where $m_\ell(z)$ acts on the shell index.

Since $\Theta$ acts only on the shell index, it has the form
$\Theta\otimes I_{\mathcal H_\ell}$ on $\mathcal H_\ell^{\oplus N}$, and hence
$$
 K_N(z)\big|_{\mathcal H_\ell^{\oplus N}}
 =
 \bigl(I_N+m_\ell(z)\Theta\bigr)\otimes I_{\mathcal H_\ell}.
$$
Thus $K_N(z)$ is not invertible if and only if $I_N+m_\ell(z)\Theta$ is not
invertible for some $\ell\ge0$, that is, if and only if
$$
 \det\bigl(I_N+m_\ell(z)\Theta\bigr)=0
$$
for some $\ell$.
Moreover,
$$
 \ker\Bigl(\bigl(I_N+m_\ell(z)\Theta\bigr)\otimes I_{\mathcal H_\ell}\Bigr)
 =
 \ker\bigl(I_N+m_\ell(z)\Theta\bigr)\otimes \mathcal H_\ell,
$$
and therefore the stated dimension formula follows.
\end{proof}

\begin{rem}
Equivalently, each $\mathcal H_\ell$ is an irreducible representation of
${\rm SO}(3)$, and every rotation--invariant operator acts as a scalar multiple
of the identity on $\mathcal H_\ell$.
In this language the $N$--shell boundary space $\mathcal H_\ell^{\oplus N}$ is a
direct sum of $N$ equivalent irreducible representations.
Here $\mathrm{Comm}(\cdot)$ denotes the commutant, that is, the algebra of all
bounded operators commuting with the ${\rm SO}(3)$--action.
By Schur's lemma, the commutant becomes isomorphic to $M_N(\mathbb C)$ acting on
the shell index and tensored with the identity on the representation space.
That is,
\[
\mathrm{Comm}\bigl(\mathcal H_\ell^{\oplus N}\bigr)
=
M_N(\mathbb C)\otimes I_{\mathcal H_\ell}.
\]
\end{rem}

\section{Variational preliminaries and higher partial waves}\label{sec:var}

From this section on we specialize to the rotationally symmetric two--shell
setting.
That is, we fix radii $0<R_1<R_2$ and take constant couplings
$\alpha_1,\alpha_2\in\mathbb R$.
We write $H:=H_2$ for the corresponding double $\delta$--shell Hamiltonian
constructed by the quadratic form method in Section~\ref{sec:form}.
The purpose of this section is to show that higher angular momentum channels
cannot produce more negative eigenvalues than the $s$--wave channel, and that
the ground state, if it exists, lies in the sector $\ell=0$.

Since the interaction is radial, $H$ commutes with rotations and admits a
partial--wave decomposition.
We first recall a standard variational comparison lemma and then apply it to
the radial quadratic forms.

\begin{lem}\label{lem:neg-eig-monotone}
Let $\mathcal X$ be a Hilbert space and let $T_0$ and $T$ be self--adjoint
operators in $\mathcal X$ which are bounded from below and have the same
closed form domain $Q$.
Assume that the spectrum of each operator below its essential spectrum is
discrete so that $\{\lambda_n(\cdot)\}_{n\ge1}$ is well-defined.
Denote their quadratic forms by $q_0$ and $q$.
Assume that
\[
q[f]\ge q_0[f],\qquad f\in Q.
\]
Then the eigenvalues of $T$ and $T_0$, counted with multiplicity and listed in
nondecreasing order, satisfy
\[
\lambda_n(T)\ge \lambda_n(T_0)\qquad\text{for all }n\ge1,
\]
and in particular
\[
N_-(T)\le N_-(T_0),
\]
where $N_-(S)$ denotes the number of negative eigenvalues of a self--adjoint
operator $S$.
\end{lem}

\begin{proof}
By the min--max principle (see e.g., \cite{KatoBook})
one has
\[
\lambda_n(T)
=
\sup_{\substack{L\subset Q\\ \dim L = n-1}}
\ \inf_{\substack{f\in Q\\ f\perp L,\,\|f\|=1}} q[f],
\qquad
\lambda_n(T_0)
=
\sup_{\substack{L\subset Q\\ \dim L = n-1}}
\ \inf_{\substack{f\in Q\\ f\perp L,\,\|f\|=1}} q_0[f].
\]
Since $q\ge q_0$ on $Q$, the inner infimum in the min--max formula satisfies
\[
\inf_{\substack{f\in Q\\ f\perp L,\,\|f\|=1}} q[f]
\ge
\inf_{\substack{f\in Q\\ f\perp L,\,\|f\|=1}} q_0[f]
\]
for every subspace $L\subset Q$ with $\dim L=n-1$.
Taking the supremum over all such $L$ gives $\lambda_n(T)\ge\lambda_n(T_0)$.

To compare the numbers of negative eigenvalues, let $N:=N_-(T_0)$.
If $N=\infty$, then the claim is trivial. Otherwise $\lambda_{N+1}(T_0)\ge 0$
by definition of $N_-(T_0)$, and hence
\[
\lambda_{N+1}(T)\ge \lambda_{N+1}(T_0)\ge 0.
\]
Therefore $T$ has at most $N$ negative eigenvalues, that is, $N_-(T)\le N_-(T_0)$.
\end{proof}

We recall the standard partial--wave decomposition of $L^2(\mathbb R^3)$ with
respect to the action of the rotation group, under which the Hamiltonian $H$
reduces to a direct sum of one--dimensional radial operators labeled by the
angular momentum $\ell$.

\begin{thm}\label{thm:partial-waves}
Fix radii $0<R_1<R_2$ and couplings $\alpha_1,\alpha_2\in\mathbb R$, and let $H$
be the double $\delta$--shell Hamiltonian constructed in Section~\ref{sec:form}.
Let $h_\ell$ denote the radial operator in the $\ell$--th partial wave.
Then for every $\ell\ge1$,
$$
\lambda_k(h_\ell)\ge \lambda_k(h_0)\qquad\text{for all }k\ge1,
$$
and consequently
$$
N_-(h_\ell)\le N_-(h_0).
$$
\end{thm}

\begin{proof}
We recall the standard partial--wave decomposition (see, for example,
\cite[Sec.~II.2]{AGHH}).
Let $\{Y_{\ell m}\}_{\ell\ge0,\,-\ell\le m\le\ell}$ be an orthonormal basis of
spherical harmonics in $L^2(S^2)$.
Then $L^2(\mathbb R^3)$ decomposes as the orthogonal sum
$$
L^2(\mathbb R^3)
=
\bigoplus_{\ell=0}^\infty
\bigoplus_{m=-\ell}^{\ell}
\mathcal H_{\ell m},
\qquad
\mathcal H_{\ell m}
=
\Bigl\{u(x)=\frac{f(r)}{r}Y_{\ell m}(\omega): f\in L^2(0,\infty)\Bigr\},
$$
with $x=r\omega$, $r=|x|$, $\omega\in S^2$.
Since the interaction is radial, $H$ commutes with rotations and leaves each
$\mathcal H_{\ell m}$ invariant.
On $\mathcal H_{\ell m}$ the operator $H$ is unitarily equivalent to a
one--dimensional radial operator $h_\ell$ acting in $L^2(0,\infty)$, whose
quadratic form is
\begin{eqnarray*}
q_\ell[f]
&=&
\int_0^\infty\Bigl(|f'(r)|^2+\frac{\ell(\ell+1)}{r^2}|f(r)|^2\Bigr)\,dr
+\alpha_1|f(R_1)|^2+\alpha_2|f(R_2)|^2,\\
&& \qquad f\in H^1(0,\infty)\ \text{with } f(0)=0.
\end{eqnarray*}
Here the point values $f(R_j)$ are well defined for $f\in H^1(0,\infty)$.
In particular, for each $\ell$ the negative eigenvalues of $H$ in the
$\ell$--th partial wave coincide, with multiplicity, with those of $h_\ell$.
Moreover, $h_\ell$ does not depend on $m$, so each eigenvalue of $h_\ell$
appears in $H$ with multiplicity $2\ell+1$.

For $\ell\ge1$ the centrifugal term is nonnegative, and therefore
$$
q_\ell[f]
=
q_0[f]+\int_0^\infty\frac{\ell(\ell+1)}{r^2}|f(r)|^2\,dr
\ge
q_0[f],\qquad f\in H^1(0,\infty).
$$
Since $q_\ell\ge q_0$ on the common form domain $H^1(0,\infty)$ with $f(0)=0$,
Lemma~\ref{lem:neg-eig-monotone} implies
$$
\lambda_k(h_\ell)\ge \lambda_k(h_0)\qquad\text{for all }k\ge1,
$$
and hence $N_-(h_\ell)\le N_-(h_0)$.
\end{proof}

\begin{rem}\label{rem:two-bound-states-reduction}
Theorem~\ref{thm:partial-waves} reduces the control of negative eigenvalues in
all partial waves to the $s$--wave channel $\ell=0$.
Once the $s$--wave problem is analyzed and $N_-(h_0)$ is bounded, the same
bound automatically holds for all $\ell\ge1$.
\end{rem}

As an immediate consequence, any negative eigenvalue of lowest energy must
occur in the $\ell=0$ channel.

\begin{cor}\label{cor:groundstate-swave}
Under the assumptions of Theorem~\ref{thm:partial-waves}, if $H$ has at least
one negative eigenvalue (cf.\ Theorem~\ref{thm:elem-res} and the trace class
property of $(H-z)^{-1}-R_0(z)$), then the lowest eigenvalue belongs to the
$s$--wave sector $\ell=0$.
\end{cor}

\begin{proof}
Let $h_\ell$ be the radial operator in the $\ell$--th partial wave.
For each $\ell\ge1$ we have $q_\ell\ge q_0$ on $H^1(0,\infty)$ with $f(0)=0$,
hence the min--max principle gives
$$
\lambda_1(h_\ell)\ge \lambda_1(h_0).
$$
Therefore the smallest negative eigenvalue among all partial waves is attained
in the $s$--wave sector $\ell=0$, and the ground state of $H$ belongs to
$\ell=0$.
\end{proof}

\section{The two--shell case: detailed analysis of the $s$--wave ($\ell=0$)}\label{sec:Eigenval}

In the two--shell case $N=2$, the resolvent framework of
Section~\ref{sec:elem-resolvent} yields the spectral condition that the boundary
operator $K(z)=I+m(z)\Theta$ fails to be invertible.
On the other hand, the same eigenvalue problem admits a direct reduction to a
one--dimensional radial ODE, leading to the usual coefficient matching
conditions in each partial wave.
The next lemma connects these two viewpoints: it reduces the operator condition
to a finite--dimensional secular equation in every angular momentum channel and
shows that this secular equation is equivalent to the radial matching
condition.

\begin{lem}\label{lem:eq-secular}
Let $m(z)=(m_{ij}(z))_{i,j=1,2}$ be the boundary integral operator on
$L^2(S^2)\oplus L^2(S^2)$ with kernel $G_z(R_i\omega,R_j\omega')$, and let
$\Theta=\mathrm{diag}(\alpha_1R_1^2,\alpha_2R_2^2)$.
Assume $\Im\sqrt{z}>0$, that is, $z\in\mathbb C\setminus[0,\infty)$ on the
principal branch.
Then the following statements are equivalent.

{\rm (i)} $z$ is an eigenvalue of $H$. Equivalently, $K(z)=I+m(z)\Theta$ is not
invertible; cf.\ Proposition~\ref{prop:eig-cond}.

{\rm (ii)} There exists $\ell\in\mathbb N\cup\{0\}$ such that
$$
\det\bigl(I+m_\ell(z)\Theta\bigr)=0,
$$
where $m_\ell(z)$ is the $2\times2$ matrix describing the action of $m(z)$ on
$\mathcal H_\ell\oplus\mathcal H_\ell$, that is,
$$
I+m_\ell(z)\Theta=
\left(
\begin{array}{cc}
1+\alpha_1R_1^2m^{(11)}_\ell(z) & \alpha_2R_2^2 m^{(12)}_\ell(z)\\[2pt]
\alpha_1R_1^2 m^{(21)}_\ell(z) & 1+\alpha_2R_2^2m^{(22)}_\ell(z)
\end{array}
\right).
$$
Moreover, for each fixed $\ell$, the condition
$\det\bigl(I+m_\ell(z)\Theta\bigr)=0$
is equivalent to the eigenvalue condition obtained from the corresponding
radial problem in the $\ell$--channel.
\end{lem}

\begin{proof}
The kernel $G_z(R_i\omega,R_j\omega')$ depends only on $\omega\cdot\omega'$ and is
therefore rotation invariant.
Hence $m(z)$ commutes with the natural action of ${\rm SO}(3)$ on
$L^2(S^2)\oplus L^2(S^2)$.
Consequently, each spherical harmonic subspace
$$
\mathcal H_\ell:=\mathrm{span}\{Y_{\ell m}:\ -\ell\le m\le\ell\}
$$
is invariant, and $m(z)$ acts on $\mathcal H_\ell\oplus\mathcal H_\ell$ as a
$2\times2$ matrix $m_\ell(z)$ on the shell index.
Equivalently, $m(z)$ is block diagonal with respect to
$$
L^2(S^2)\oplus L^2(S^2)
=
\bigoplus_{\ell=0}^{\infty}\bigoplus_{m=-\ell}^{\ell}
\bigl(\mathbb C\,Y_{\ell m}\oplus \mathbb C\,Y_{\ell m}\bigr),
$$
and on each block it acts as
$$
m_\ell(z)=
\left(
\begin{array}{cc}
m^{(11)}_\ell(z) & m^{(12)}_\ell(z)\\[2pt]
m^{(21)}_\ell(z) & m^{(22)}_\ell(z)
\end{array}
\right).
$$

Therefore $K(z)=I+m(z)\Theta$ is unitarily equivalent to the orthogonal direct sum of
the matrices $I+m_\ell(z)\Theta$, where each block appears with multiplicity
$2\ell+1$.
Since $m(z)$ is compact by Lemma~\ref{lem:compact-m}, the operator $K(z)$ is
Fredholm of index zero.

If $\det(I+m_\ell(z)\Theta)\ne 0$ for every $\ell$, then each block
$I+m_\ell(z)\Theta$ is invertible and hence $\ker K(z)=\{0\}$.
For a Fredholm operator of index zero, injectivity implies surjectivity, so
$K(z)$ is invertible.
Conversely, if $\det(I+m_\ell(z)\Theta)=0$ for some $\ell$, then the
corresponding block has a nontrivial kernel, and so does $K(z)$.
Thus $K(z)$ is not invertible if and only if there exists $\ell$ such that
$\det(I+m_\ell(z)\Theta)=0$, which proves the equivalence of (i) and (ii).

We now relate the condition $\det(I+m_\ell(z)\Theta)=0$ to the eigenvalue
condition arising from the coefficient matching in the radial problem.
Let $g={}^{t}(g_1,g_2)\ne0$ satisfy $(I+m_\ell(z)\Theta)g=0$ and define
$$
u = -\sum_{j=1}^{2}\Gamma_j(z)\,(\Theta g)_j,
\qquad
(\Theta g)_j=\alpha_j R_j^2 g_j .
$$
Then $u$ solves $(-\Delta-z)u=0$ away from $S_1\cup S_2$.
Since $g_j\in\mathcal H_\ell$, the angular dependence of $u$ lies in
$\mathcal H_\ell$ and is proportional to $Y_{\ell m}$ for some $m$.
We represent $u$ in the form
\[
u(x)=\frac{u_\ell(r)}{r}\,Y_{\ell m}(\omega),\qquad x=r\omega.
\]
The mapping properties of $\Gamma_j(z)$ then imply that on each interval
$(0,R_1)$, $(R_1,R_2)$, and $(R_2,\infty)$ the function $u_\ell$ is a linear
combination of the standard radial solutions of $(-\Delta-z)u=0$, expressed in
terms of spherical Bessel and Hankel functions.

Taking traces on $S_1$ and $S_2$ and using the $\delta$--shell jump laws,
we recall that $u$ is continuous at $r=R_j$ and satisfies
\begin{equation*}
\partial_r u(R_j+0,\omega)-\partial_r u(R_j-0,\omega)
=
\alpha_ju(R_j,\omega),
\qquad j=1,2.
\end{equation*}
Moreover, by the trace identity $\tau_i\Gamma_j(z)=m_{ij}(z)$ and the standard
jump relations for single--layer potentials (cf. Lemma~\ref{lem:compact-m} and
the proof of Proposition~\ref{prop:eig-cond}),
the resulting matching conditions yield a homogeneous linear system for the
$\ell$--channel boundary data.
By the definition of $m_\ell(z)$ (see Appendix~\ref{app:ml}, in particular
Lemma~\ref{lem:ml-formula}), this boundary system is exactly
$(I+m_\ell(z)\Theta)g=0$.
Thus a nontrivial vector $g$ in $\ker(I+m_\ell(z)\Theta)$ produces a nontrivial
radial solution satisfying the matching conditions, and conversely every such
radial solution produces a nontrivial $g$.
Hence $\det(I+m_\ell(z)\Theta)=0$ is equivalent to the radial matching secular
equation in the $\ell$--channel.
\end{proof}

We now specialize to the $s$--wave sector $\ell=0$ and derive an explicit
eigenvalue condition by solving the corresponding radial problem.
\begin{prop}\label{prop:s-wave-secular}
Let $0<R_1<R_2$, $\alpha_1,\alpha_2\in\mathbb R$, and set $d=R_2-R_1$.
In the $s$--wave sector ($\ell=0$) we write $u(r)=r f(r)$.
Then $u$ satisfies
\begin{equation}\label{1D}
-\,u''(r)=E\,u(r)\quad \text{for } r\ne R_1,R_2,\qquad
u(0)=0,\qquad u\in L^2(0,\infty),
\end{equation}
together with the interface conditions
\begin{equation}\label{u-jump}
u \ \text{is continuous at } R_j,\qquad
u'(R_j+0)-u'(R_j-0)=\alpha_j\,u(R_j),\quad j=1,2.
\end{equation}
For $E=-\kappa^2<0$ the $s$--wave eigenvalues are precisely those numbers
$E=-\kappa^2$ with $\kappa>0$ which satisfy
\begin{equation}\label{eq:sec}
F_d(\kappa):=\Bigl[\gamma_1(\kappa)+(\alpha_2+\kappa)\Bigr]\cosh(\kappa d)
+\Bigl[\kappa+\gamma_1(\kappa)+\frac{\alpha_2\,\gamma_1(\kappa)}{\kappa}\Bigr]\sinh(\kappa d)=0,
\end{equation}
where
\begin{equation}\label{gamma1}
\gamma_1(\kappa)=\alpha_1+\kappa\,\coth(\kappa R_1).
\end{equation}
Equivalently, nontrivial solutions exist if and only if $\det M(\kappa;d)=0$,
where
\begin{equation}\label{eq:matrix}
M(\kappa;d)=
\left(
\begin{array}{cc}
-\gamma_1(\kappa) & \kappa\\[4pt]
(\alpha_2+\kappa)\cosh(\kappa d)+\kappa\sinh(\kappa d) &
(\alpha_2+\kappa)\sinh(\kappa d)+\kappa\cosh(\kappa d)
\end{array}
\right).
\end{equation}
\end{prop}

\begin{proof}
We take $E=-\kappa^2<0$ with $\kappa>0$.
The general solution of \eqref{1D} in the three regions $0<r<R_1$,
$R_1<r<R_2$, and $r>R_2$ can be written as
\begin{equation}\label{u-parts}
u(r)=
\begin{cases}
A\,\sinh(\kappa r), & 0<r<R_1,\\[1ex]
B\,e^{\kappa r}+C\,e^{-\kappa r}, & R_1<r<R_2,\\[1ex]
D\,e^{-\kappa r}, & r>R_2,
\end{cases}
\end{equation}
with constants $A,B,C,D\in\mathbb C$.
Continuity at $R_1$ and $R_2$ gives
\begin{eqnarray}\label{cont}
A\sinh(\kappa R_1) &=& B e^{\kappa R_1}+C e^{-\kappa R_1},\\
D e^{-\kappa R_2} &=& B e^{\kappa R_2}+C e^{-\kappa R_2}.
\end{eqnarray}
Differentiating \eqref{u-parts} yields
\begin{eqnarray}\label{derivs}
u_1'(R_1) &=& A\kappa\cosh(\kappa R_1),\nonumber\\
u_2'(R_1) &=& \kappa\bigl(B e^{\kappa R_1}-C e^{-\kappa R_1}\bigr),\nonumber\\
u_2'(R_2) &=& \kappa\bigl(B e^{\kappa R_2}-C e^{-\kappa R_2}\bigr),\nonumber\\
u_3'(R_2) &=& -\,\kappa D e^{-\kappa R_2}.
\end{eqnarray}
Using the jump conditions \eqref{u-jump} at $R_1$ and $R_2$ we obtain
\begin{eqnarray}\label{jumps-raw}
\kappa\bigl(B e^{\kappa R_1}-C e^{-\kappa R_1}\bigr)-A\kappa\cosh(\kappa R_1)
&=& \alpha_1\bigl(B e^{\kappa R_1}+C e^{-\kappa R_1}\bigr),\\
-\kappa D e^{-\kappa R_2}-\kappa\bigl(B e^{\kappa R_2}-C e^{-\kappa R_2}\bigr)
&=& \alpha_2\bigl(B e^{\kappa R_2}+C e^{-\kappa R_2}\bigr).
\end{eqnarray}
From \eqref{cont} we eliminate $A$ and $D$ and obtain
\begin{equation}\label{A-D}
A=\frac{B e^{\kappa R_1}+C e^{-\kappa R_1}}{\sinh(\kappa R_1)},\qquad
D e^{-\kappa R_2}=B e^{\kappa R_2}+C e^{-\kappa R_2}.
\end{equation}
Substituting \eqref{A-D} into \eqref{jumps-raw} yields
\begin{eqnarray}\label{two-eqs}
\kappa\bigl(B e^{\kappa R_1}-C e^{-\kappa R_1}\bigr)
&=&\bigl(\alpha_1+\kappa\coth(\kappa R_1)\bigr)\bigl(B e^{\kappa R_1}+C e^{-\kappa R_1}\bigr),\\
\kappa\bigl(B e^{\kappa R_2}-C e^{-\kappa R_2}\bigr)
&=&-\bigl(\alpha_2+\kappa\bigr)\bigl(B e^{\kappa R_2}+C e^{-\kappa R_2}\bigr).
\end{eqnarray}

We now introduce the combinations
\begin{equation}\label{XY-def}
X(r)=B e^{\kappa r}+C e^{-\kappa r},\qquad
Y(r)=B e^{\kappa r}-C e^{-\kappa r}.
\end{equation}
Then \eqref{two-eqs} can be rewritten as
\begin{equation}\label{XY-BC}
\kappa Y(R_1)=\gamma_1(\kappa) X(R_1),\qquad
\kappa Y(R_2)=-(\alpha_2+\kappa)X(R_2),
\end{equation}
where $\gamma_1(\kappa)$ is given by \eqref{gamma1}.
Between $R_1$ and $R_2$ the pair $(X,Y)$ satisfies
\begin{eqnarray}
X(R_2)&=&\cosh(\kappa d)\,X(R_1)+\sinh(\kappa d)\,Y(R_1), \label{eq:propX}\\
Y(R_2)&=&\sinh(\kappa d)\,X(R_1)+\cosh(\kappa d)\,Y(R_1). \label{eq:propY}
\end{eqnarray}
From \eqref{XY-BC} at $R_1$ we have
\begin{equation}\label{Y1}
Y(R_1)=\frac{\gamma_1(\kappa)}{\kappa}\,X(R_1).
\end{equation}
Substituting \eqref{Y1} into \eqref{eq:propX} and \eqref{eq:propY} yields
\begin{eqnarray}\label{prop2}
X(R_2)&=&\Bigl(\cosh(\kappa d)+\frac{\gamma_1(\kappa)}{\kappa}\sinh(\kappa d)\Bigr)X(R_1),\\
Y(R_2)&=&\Bigl(\sinh(\kappa d)+\frac{\gamma_1(\kappa)}{\kappa}\cosh(\kappa d)\Bigr)X(R_1).
\end{eqnarray}
Applying the boundary condition at $R_2$ in \eqref{XY-BC}, we obtain a homogeneous
linear system for $X(R_1)$ and $Y(R_1)$,
$$
M(\kappa;d)
\left(
\begin{array}{c}
X(R_1)\\[2pt]
Y(R_1)
\end{array}
\right)
=0,
$$
with $M(\kappa;d)$ as in \eqref{eq:matrix}.
There exists a nontrivial solution if and only if $\det M(\kappa;d)=0$.
A direct computation of the determinant gives \eqref{eq:sec}, and the statement follows.
\end{proof}

\section{Counting of $s$--wave bound states for large shell separation}\label{sec:counting}

In this section we restrict to the $s$--wave sector $\ell=0$ and study how the
number of negative eigenvalues behaves when the shell separation
$d=R_2-R_1$ becomes large.
When the two $\delta$--shells are far apart, bound states originating from
attractive single shells interact only weakly and the coupling is exponentially
small in $d$.
As a result, the number of negative $s$--wave eigenvalues stabilizes and can be
read off from the corresponding one--shell problems.

\begin{prop}\label{prop:neg-crit-corrected}
Let $H$ be the double $\delta$--shell Hamiltonian with radii $0<R_1<R_2$ and
couplings $\alpha_1,\alpha_2\in\mathbb R$.
In the $s$--wave subspace, the negative eigenvalues $E=-\kappa^2$ are in
one--to--one correspondence with the positive roots $\kappa>0$ of the secular
equation~\eqref{eq:sec}.

The associated $s$--wave quadratic form is given by
\[
q_0[f]
=
\int_0^\infty |f'(r)|^2\,dr
+\alpha_1|f(R_1)|^2+\alpha_2|f(R_2)|^2,
\qquad f\in H^1(0,\infty).
\]
Hence the negative part of $q_0$ depends only on the two point evaluations
$f\mapsto f(R_1)$ and $f\mapsto f(R_2)$.
As a consequence, the $s$--wave radial operator can have at most two negative
eigenvalues (counted with multiplicity), and therefore the secular
equation~\eqref{eq:sec} admits at most two positive roots.
\end{prop}

\begin{proof}
By Proposition~\ref{prop:s-wave-secular} the negative $s$--wave eigenvalues
$E=-\kappa^2$ are characterized by the positive roots $\kappa>0$ of
$F_d(\kappa)=0$, where $F_d$ is given in \eqref{eq:sec}.
We therefore study the function $F_d(\kappa)$.

The $s$--wave secular equation has the form
$$
F_d(\kappa)
=
A(\kappa)\cosh(\kappa d)+B(\kappa)\sinh(\kappa d),
\qquad \kappa>0,
$$
where
\begin{eqnarray*}
A(\kappa)&=&\gamma_1(\kappa)+\alpha_2+\kappa,\\
B(\kappa)&=&\kappa+\gamma_1(\kappa)+\frac{\alpha_2\gamma_1(\kappa)}{\kappa},\\
\gamma_1(\kappa)&=&\alpha_1+\kappa\coth(\kappa R_1).
\end{eqnarray*}
Using $\cosh t=\tfrac12(e^t+e^{-t})$ and $\sinh t=\tfrac12(e^t-e^{-t})$, we obtain
\begin{equation}\label{eq:Fd-split-new}
F_d(\kappa)=\tfrac12\bigl[F_\infty(\kappa)e^{\kappa d}+G(\kappa)e^{-\kappa d}\bigr],
\end{equation}
where
\begin{equation}\label{eq:Finfty-G}
F_\infty(\kappa)=\frac{(\kappa+\gamma_1(\kappa))(2\kappa+\alpha_2)}{\kappa},\qquad
G(\kappa)=\alpha_2\Bigl(1-\frac{\gamma_1(\kappa)}{\kappa}\Bigr).
\end{equation}

\smallskip
\noindent
We first show that the $s$--wave radial operator has at most two negative
eigenvalues.
Assume for contradiction that it has at least three negative eigenvalues and
let $L$ be the span of three corresponding eigenfunctions.
Then $\dim L=3$ and
$$
q_0[f]=(h_0f,f)_{L^2(0,\infty)}<0\qquad \text{for all }0\ne f\in L.
$$
The linear map $f\mapsto (f(R_1),f(R_2))\in\mathbb C^2$ cannot be injective on
$L$, hence there exists $0\ne f\in L$ such that $f(R_1)=f(R_2)=0$.
For this $f$ we have
$$
q_0[f]=\int_0^\infty |f'(r)|^2\,dr>0,
$$
since $\int_0^\infty |f'(r)|^2\,dr=0$ would imply that $f$ is constant, and the
conditions $f(R_1)=f(R_2)=0$ would force $f\equiv0$.
This contradiction shows $N_-(h_0)\le2$, and therefore $F_d(\kappa)=0$ has at
most two positive roots, counted with multiplicity.

\smallskip
\noindent
(i) We first consider the case $\alpha_1,\alpha_2\ge0$.
Then $\gamma_1(\kappa)\ge0$, $A(\kappa)>0$, and $B(\kappa)\ge0$ for all
$\kappa>0$.
Since $\cosh(\kappa d)>0$ and $\sinh(\kappa d)>0$ for all $\kappa>0$ and $d>0$,
we have $F_d(\kappa)>0$ for all $\kappa>0$.
Hence, there is no positive root of $F_d$, and therefore no negative $s$--wave
eigenvalue.

\smallskip
\noindent
(ii) Next we assume that exactly one shell is attractive and that the
corresponding one--shell problem supports an $s$--wave bound state.

\smallskip
\noindent
Case 1: $\alpha_1<0\le\alpha_2$ and $\alpha_1<-1/R_1$.
Then, there exists a unique $\kappa_{\mathrm{in}}>0$ satisfying
$$
\kappa+\gamma_1(\kappa)=0,
\qquad\text{that is,}\qquad
\alpha_1+\kappa\coth(\kappa R_1)+\kappa=0.
$$
At $\kappa=\kappa_{\mathrm{in}}$ one has $F_\infty(\kappa_{\mathrm{in}})=0$.
Moreover $F_\infty'(\kappa_{\mathrm{in}})\ne0$.
Indeed, $2\kappa_{\mathrm{in}}+\alpha_2>0$ and
$\kappa_{\mathrm{in}}+\gamma_1(\kappa_{\mathrm{in}})=0$ imply
$F_\infty'(\kappa_{\mathrm{in}})>0$.
Expanding \eqref{eq:Fd-split-new} near $\kappa_{\mathrm{in}}$ yields
$$
\kappa(d)=\kappa_{\mathrm{in}}
-\frac{G(\kappa_{\mathrm{in}})}{F'_\infty(\kappa_{\mathrm{in}})}e^{-2\kappa_{\mathrm{in}}d}
+O(e^{-4\kappa_{\mathrm{in}}d}),
$$
so for all sufficiently large $d$ there exists exactly one positive root of
$F_d$.
The corresponding eigenvalue gives the unique negative $s$--wave eigenvalue in
this case.

\smallskip
\noindent
Case 2: $\alpha_2<0\le\alpha_1$ and $\alpha_2<-1/R_2$.
For the single outer shell at $r=R_2$, the $s$--wave bound state condition is
$\alpha_2<-1/R_2$, and the corresponding root is given by the unique positive
zero of $2\kappa+\alpha_2$, namely
$$
\kappa_{\mathrm{out}}=-\frac{\alpha_2}{2}>0.
$$
At $\kappa=\kappa_{\mathrm{out}}$ we have $F_\infty(\kappa_{\mathrm{out}})=0$.
Moreover
$$
F_\infty'(\kappa_{\mathrm{out}})
=\frac{2\bigl(\kappa_{\mathrm{out}}+\gamma_1(\kappa_{\mathrm{out}})\bigr)}{\kappa_{\mathrm{out}}}
\ne0,
$$
and since $\alpha_1\ge0$ implies $\gamma_1(\kappa_{\mathrm{out}})>0$, we have
$F_\infty'(\kappa_{\mathrm{out}})>0$.
Expanding \eqref{eq:Fd-split-new} near $\kappa_{\mathrm{out}}$ yields
$$
\kappa(d)=\kappa_{\mathrm{out}}
-\frac{G(\kappa_{\mathrm{out}})}{F'_\infty(\kappa_{\mathrm{out}})}e^{-2\kappa_{\mathrm{out}}d}
+O(e^{-4\kappa_{\mathrm{out}}d}).
$$
Thus, we again obtain exactly one positive root for all sufficiently large $d$.
This root gives the unique negative $s$--wave eigenvalue in this case.

\smallskip
\noindent
(iii) Finally we consider the case when both shells are attractive,
$\alpha_1<0$ and $\alpha_2<0$.
Assume that each one--shell subsystem supports an $s$--wave bound state, that is,
$\alpha_j<-1/R_j$ for $j=1,2$.
Let $\kappa_{\mathrm{in}},\kappa_{\mathrm{out}}>0$ be the corresponding isolated
roots, that is, the zeros of $\kappa+\gamma_1(\kappa)$ and $2\kappa+\alpha_2$,
respectively.
Then
$$
F_\infty(\kappa_{\mathrm{in}})=0,\qquad F_\infty(\kappa_{\mathrm{out}})=0.
$$
Assume that these zeros are simple, that is,
$F_\infty'(\kappa_{\mathrm{in}})\ne0$ and $F_\infty'(\kappa_{\mathrm{out}})\ne0$.
(If one of them is not simple, then the corresponding critical tuning is treated
separately in the tunneling regime.)
For large $d$ the roots are only slightly perturbed, and we can write
\begin{eqnarray}
\delta_{\mathrm{in}}(d)&=&-\frac{G(\kappa_{\mathrm{in}})}{F'_\infty(\kappa_{\mathrm{in}})}e^{-2\kappa_{\mathrm{in}}d}
+O(e^{-4\kappa_{\mathrm{in}}d}),\nonumber\\
\delta_{\mathrm{out}}(d)&=&-\frac{G(\kappa_{\mathrm{out}})}{F'_\infty(\kappa_{\mathrm{out}})}e^{-2\kappa_{\mathrm{out}}d}
+O(e^{-4\kappa_{\mathrm{out}}d}),\nonumber
\end{eqnarray}
so for all sufficiently large $d$ there are two distinct positive roots of $F_d$,
namely $\kappa_{\mathrm{in}}+\delta_{\mathrm{in}}(d)$ and
$\kappa_{\mathrm{out}}+\delta_{\mathrm{out}}(d)$.

For general values of $d$, the number of positive roots in $(0,\infty)$ can
change only when $F_d$ has a multiple root $\kappa_0>0$, that is,
\begin{equation}\label{eq:double-root-new}
\begin{cases}
A(\kappa_0)\cosh(\kappa_0d)+B(\kappa_0)\sinh(\kappa_0d)=0,\\
(A'(\kappa_0)+dB(\kappa_0))\cosh(\kappa_0d)
+(B'(\kappa_0)+dA(\kappa_0))\sinh(\kappa_0d)=0.
\end{cases}
\end{equation}
We exclude here the threshold $\kappa=0$.
This gives a codimension--one condition on the parameters $(\alpha_1,\alpha_2,d)$.
In the generic situation the two roots persist and merge only at isolated
critical values of $(\alpha_1,\alpha_2,d)$.
\end{proof}

\begin{cor}\label{cor:no-bound-fixed}
Assume that exactly one of $\alpha_1,\alpha_2$ is negative and that the
corresponding single $\delta$--shell does not support an $s$--wave bound state,
that is,
$$
  \alpha_1 \ge -\frac{1}{R_1}\ \text{if }\alpha_2\ge0,
  \qquad\text{or}\qquad
  \alpha_2 \ge -\frac{1}{R_2}\ \text{if }\alpha_1\ge0.
$$
Then the $s$--wave secular equation~\eqref{eq:sec} has no positive root.
Consequently, $H$ has no negative eigenvalue.
\end{cor}

\begin{proof}
We first show that the $s$--wave secular equation has no positive root.

\smallskip
\noindent
Case 1: $\alpha_1<0\le\alpha_2$ and $\alpha_1\ge-1/R_1$.
For $\kappa>0$ one has $\kappa\coth(\kappa R_1)\ge 1/R_1$, hence
$\gamma_1(\kappa)=\alpha_1+\kappa\coth(\kappa R_1)\ge0$.
Therefore
$$
A(\kappa)=\gamma_1(\kappa)+\alpha_2+\kappa>0,
\qquad
B(\kappa)=\kappa+\gamma_1(\kappa)+\frac{\alpha_2\gamma_1(\kappa)}{\kappa}\ge0,
$$
and since $\cosh(\kappa d)>0$ and $\sinh(\kappa d)>0$ for $d>0$ we obtain
$F_d(\kappa)>0$ for all $\kappa>0$.
Thus \eqref{eq:sec} has no positive root.

\smallskip
\noindent
Case 2: $\alpha_2<0\le\alpha_1$ and $\alpha_2\ge-1/R_2$.
Let $q_0$ be the $s$--wave quadratic form,
$$
q_0[f]=\int_0^\infty |f'(r)|^2\,dr+\alpha_1|f(R_1)|^2+\alpha_2|f(R_2)|^2,
\qquad f\in H^1(0,\infty).
$$
Since $\alpha_1\ge0$, we have
$$
q_0[f]
\ge
\int_0^\infty |f'(r)|^2\,dr+\alpha_2|f(R_2)|^2
=:q_{\mathrm{out}}[f].
$$
The form $q_{\mathrm{out}}$ is the $s$--wave form for the single outer shell at
$r=R_2$, and the hypothesis $\alpha_2\ge-1/R_2$ means that this one--shell
problem has no negative $s$--wave eigenvalue.
By Lemma~\ref{lem:neg-eig-monotone} it follows that the double--shell $s$--wave
operator $h_0$ also has no negative eigenvalue.
Hence \eqref{eq:sec} has no positive root in this case.

\smallskip
\noindent
Having excluded negative eigenvalues in the $s$--wave channel, we treat higher
partial waves.
For $\ell\ge1$ the radial quadratic form is
$$
q_\ell[f]
=
\int_0^\infty |f'(r)|^2\,dr
+ \ell(\ell+1)\int_0^\infty \frac{|f(r)|^2}{r^2}\,dr
+\alpha_1|f(R_1)|^2+\alpha_2|f(R_2)|^2,
$$
so $q_\ell[f]\ge q_0[f]$ for all $f\in H^1(0,\infty)$.
Another application of Lemma~\ref{lem:neg-eig-monotone} gives
$N_-(h_\ell)\le N_-(h_0)=0$ for every $\ell\ge1$.
Therefore $H$ has no negative eigenvalue.
\end{proof}

\begin{rem}\label{rem:outer-attractive-large-d}
In the outer--attractive case $\alpha_2<0\le\alpha_1$, the one--shell threshold
$-1/R_2$ tends to $0^-$ as $R_2=R_1+d\to\infty$.
Hence any fixed $\alpha_2<0$ eventually satisfies $\alpha_2<-1/R_2$, and for all
sufficiently large $d$ the $s$--wave secular equation acquires a positive root
and produces a negative $s$--wave eigenvalue as in
Proposition~\ref{prop:neg-crit-corrected}(ii).
\end{rem}

\section{Tunneling Splitting of Eigenvalues}\label{sec:tunneling}

In this section we analyze the tunneling splitting of the lowest $s$--wave
eigenvalues for the double $\delta$--shell operator in the regime of large shell
separation.
As shown in Section~\ref{sec:counting}, when the two shells are well separated
and the corresponding one--shell eigenvalues are distinct, each $s$--wave
eigenvalue converges to its one--shell counterpart with an exponentially small
correction of order $e^{-2\kappa d}$.

To isolate a genuine tunneling effect, we tune the parameters so that the
limiting one--shell eigenvalues coincide at a common energy
$E_0=-\kappa_0^2<0$.
In this critical situation, the standard perturbative picture breaks down and
the interaction between the two shells produces a pair of eigenvalues whose
splitting occurs on the larger scale $e^{-\kappa_0 d}$.
This regime is analogous to the symmetric double--well situation in
one--dimensional quantum mechanics and provides a natural setting in which to
compare the exact splitting with the prediction based on Agmon distances.
We now establish this splitting explicitly within the present solvable model.

\begin{thm}[Tunneling splitting of $s$--wave eigenvalues]\label{thm:tunnel}
Let $R_1>0$ and $\alpha_1,\alpha_2\in\mathbb R$.
Assume that there exists $\kappa_0>0$ such that
\begin{equation}\label{eq:tuning}
\gamma_1(\kappa_0)+\kappa_0=0,\qquad \alpha_2+2\kappa_0=0,
\end{equation}
where $\gamma_1(\kappa)=\alpha_1+\kappa\coth(\kappa R_1)$.
Then, as $d=R_2-R_1\to\infty$, the secular equation admits two solutions
$\kappa_\pm(d)$ near $\kappa_0$, and they satisfy
\begin{equation}\label{eq:kappa-split}
\kappa_\pm(d)=\kappa_0\pm C\,e^{-\kappa_0 d}+o(e^{-\kappa_0 d}),
\end{equation}
with the constant
\begin{equation}\label{eq:C-constant}
C=\left(-\,\frac{2\,G(\kappa_0)}{F_\infty''(\kappa_0)}\right)^{1/2}>0 .
\end{equation}
Consequently,
\begin{equation}\label{eq:E-split}
E_\pm(d)=-\kappa_\pm(d)^2,\qquad
\bigl|E_+(d)-E_-(d)\bigr|
=4\kappa_0C\,e^{-\kappa_0 d}+o(e^{-\kappa_0 d}) .
\end{equation}
\end{thm}

\begin{proof}
We use the representation \eqref{eq:Fd-split-new},
$$
F_d(\kappa)=\tfrac12\bigl[F_\infty(\kappa)e^{\kappa d}+G(\kappa)e^{-\kappa d}\bigr],
\qquad
F_\infty(\kappa)=\frac{(\kappa+\gamma_1(\kappa))(2\kappa+\alpha_2)}{\kappa},
$$
where $G$ is given in \eqref{eq:Finfty-G}.
The tuning conditions \eqref{eq:tuning} imply
$$
\kappa_0+\gamma_1(\kappa_0)=0,\qquad 2\kappa_0+\alpha_2=0,
$$
and therefore $F_\infty(\kappa_0)=0$.

Differentiating $F_\infty$ and using the fact that both factors vanish at
$\kappa_0$, we obtain $F_\infty'(\kappa_0)=0$, so $\kappa_0$ is at least a
double zero of $F_\infty$.
Moreover,
$$
\gamma_1'(\kappa)=\coth(\kappa R_1)-\kappa R_1\,\mathrm{csch}^2(\kappa R_1)
=\frac{\sinh(2t)-2t}{2\sinh^2 t}>0,
\qquad t=\kappa R_1>0,
$$
and hence
$$
F_\infty''(\kappa_0)=\frac{4\,(1+\gamma_1'(\kappa_0))}{\kappa_0}>0.
$$
Furthermore, by \eqref{eq:Finfty-G} and $2\kappa_0+\alpha_2=0$ we have
$$
G(\kappa_0)=\alpha_2\Bigl(1-\frac{\gamma_1(\kappa_0)}{\kappa_0}\Bigr)
=\alpha_2\Bigl(1-(-1)\Bigr)=2\alpha_2=-4\kappa_0<0.
$$
Therefore
$$
-\frac{2\,G(\kappa_0)}{F_\infty''(\kappa_0)}>0,
$$
and the constant $C$ in \eqref{eq:C-constant} is well defined and strictly
positive.
Since $F_\infty(\kappa_0)=F_\infty'(\kappa_0)=0$ and $F_\infty''(\kappa_0)>0$,
the zero of $F_\infty$ at $\kappa_0$ has multiplicity exactly two.

\medskip
\noindent
\emph{Existence of two roots near $\kappa_0$ for large $d$.}
Fix $\delta_0\in(0,\kappa_0/2)$ so that $F_\infty$ and $G$ are $C^3$ on
$[\kappa_0-\delta_0,\kappa_0+\delta_0]$.
Then the Taylor remainder
\begin{equation}\label{eq:Finf-unif-Taylor}
F_\infty(\kappa_0+\delta)
=\tfrac12\,F_\infty''(\kappa_0)\,\delta^2+O(\delta^3)
\end{equation}
holds uniformly for $|\delta|\le\delta_0$, and
$G(\kappa)=G(\kappa_0)+O(\delta)$ holds uniformly on the same interval.

Write $\kappa=\kappa_0+\delta$ and introduce the scaling
$$
\delta=\eta e^{-\kappa_0 d},\qquad |\eta|\le M,
$$
where $M>0$ is fixed.
For all sufficiently large $d$ we have $|\delta|\le\delta_0$.
From \eqref{eq:Fd-split-new} we can rewrite $F_d(\kappa)=0$ as
\begin{equation}\label{eq:scaled-balance}
F_\infty(\kappa)+G(\kappa)e^{-2\kappa d}=0.
\end{equation}
Multiplying \eqref{eq:scaled-balance} by $e^{2\kappa_0 d}$ and using
$\delta d=\eta d e^{-\kappa_0 d}=o(1)$ uniformly for $|\eta|\le M$, we obtain
$$
e^{2\kappa_0 d}F_\infty(\kappa_0+\eta e^{-\kappa_0 d})
+G(\kappa_0)e^{-2\delta d}
+o(1)
=0
\qquad(d\to\infty),
$$
where the $o(1)$ term is uniform for $|\eta|\le M$.
Using \eqref{eq:Finf-unif-Taylor} with $\delta=\eta e^{-\kappa_0 d}$ and
$e^{-2\delta d}=1+o(1)$ uniformly for $|\eta|\le M$, we arrive at
\begin{equation}\label{eq:H-limit}
H_d(\eta)
:=e^{2\kappa_0 d}\,F_\infty(\kappa_0+\eta e^{-\kappa_0 d})+G(\kappa_0)
=
Q(\eta)+r_d(\eta),
\end{equation}
where
$$
Q(\eta):=\tfrac12F_\infty''(\kappa_0)\eta^2+G(\kappa_0),
\qquad
\sup_{|\eta|\le M}|r_d(\eta)|\to0\quad(d\to\infty).
$$
The polynomial $Q$ has exactly two simple real zeros $\eta=\pm C$, where
$C=\bigl(-2G(\kappa_0)/F_\infty''(\kappa_0)\bigr)^{1/2}$.
Choose $\varepsilon\in(0,C)$ so that $Q(C-\varepsilon)$ and $Q(C+\varepsilon)$
have opposite signs, and likewise $Q(-C-\varepsilon)$ and $Q(-C+\varepsilon)$
have opposite signs.
By \eqref{eq:H-limit} and the uniform convergence $H_d\to Q$ on the two compact
intervals $[C-\varepsilon,C+\varepsilon]$ and $[-C-\varepsilon,-C+\varepsilon]$,
the same sign pattern holds for $H_d$ for all sufficiently large $d$.
Hence, by the intermediate value theorem, there exist
$$
\eta_+(d)\in(C-\varepsilon,C+\varepsilon),\qquad
\eta_-(d)\in(-C-\varepsilon,-C+\varepsilon)
$$
such that $H_d(\eta_\pm(d))=0$.

To show $\eta_\pm(d)\to\pm C$, take any sequence $d_n\to\infty$.
Since $\eta_+(d_n)\in[C-\varepsilon,C+\varepsilon]$, there exists a subsequence
$\eta_+(d_{n_k})\to\eta_*$ with $\eta_*\in[C-\varepsilon,C+\varepsilon]$.
Because $H_{d_{n_k}}(\eta_+(d_{n_k}))=0$ and $H_{d_{n_k}}\to Q$ uniformly on
$[C-\varepsilon,C+\varepsilon]$, we obtain $Q(\eta_*)=0$, hence $\eta_*=C$.
Thus $\eta_+(d)\to C$, and similarly $\eta_-(d)\to -C$.

Define
$$
\kappa_\pm(d):=\kappa_0+\eta_\pm(d)e^{-\kappa_0 d},
\qquad
\delta_\pm(d):=\kappa_\pm(d)-\kappa_0.
$$
Then $F_d(\kappa_\pm(d))=0$ and $\eta_\pm(d)=\pm C+o(1)$.
Therefore
$$
\delta_\pm(d)=\pm\,C\,e^{-\kappa_0 d}+o(e^{-\kappa_0 d}),
$$
which proves \eqref{eq:kappa-split}.

Finally,
$$
E_\pm(d)=-\kappa_\pm(d)^2
=-\bigl(\kappa_0+\delta_\pm(d)\bigr)^2
=-\kappa_0^2-2\kappa_0\,\delta_\pm(d)+O(\delta_\pm(d)^2),
$$
and hence
$$
E_+(d)-E_-(d)
=-2\kappa_0\bigl(\delta_+(d)-\delta_-(d)\bigr)+o(e^{-\kappa_0 d})
=-4\kappa_0C\,e^{-\kappa_0 d}+o(e^{-\kappa_0 d}),
$$
which is equivalent to \eqref{eq:E-split}.
\end{proof}

\begin{rem}[Agmon distance heuristic]\label{rem:agmon}
It is instructive to compare the exponential factor in \eqref{eq:E-split} with
the standard tunneling heuristics based on Agmon metrics.
For a Schr\"odinger operator $-\Delta+V(x)$ and an energy $E<\inf V$, the Agmon
distance between two points $x$ and $y$ is defined by
$$
d_E(x,y) = \inf_{\gamma}\int_0^1
\sqrt{(V(\gamma(t))-E)_+}\,|\dot\gamma(t)|\,dt ,
\qquad
(\cdot)_+ = \max(\cdot,0).
$$
In the present model the potential vanishes in the gap region, so $V=0$ between
the shells, while $E=-\kappa_0^2<0$.
Hence $V-E=\kappa_0^2$ in the gap and the integrand becomes the constant
$\sqrt{-E}=\kappa_0$.
Therefore the Agmon action across $(R_1,R_2)$ equals
$$
d_E(R_1,R_2)=\int_{R_1}^{R_2}\kappa_0\,dr=\kappa_0 d,
$$
and the splitting scale $\exp(-d_E(R_1,R_2))=\exp(-\kappa_0 d)$ agrees with the
explicit asymptotics in Theorem~\ref{thm:tunnel}.
Here $\exp(-2d_E)$ corresponds to a tunneling probability (a squared amplitude),
while the energy splitting is proportional to an overlap amplitude and
therefore has the scale $\exp(-d_E)$.
\end{rem}

\section*{Acknowledgements}

The author is grateful to Professor Pavel Exner for drawing his attention
to earlier results on concentric $\delta$--shell interactions obtained by
J. Shabani and collaborators.
This comment led to a substantial clarification of the historical context
and to an improvement of the introduction.

\section*{Appendix}

\appendix

\section{Calibration with semiconductor quantum dots}\label{sec:experiment}
This appendix provides a scale and trend check for the double $\delta$--shell model
based on representative parameters for core--shell quantum dots.
It is not intended as a quantitative fit to any specific sample.
We focus on the conduction--band profile for the electron.
Optical transition energies, however, depend on additional ingredients such as
the valence--band structure, excitonic Coulomb attraction, and polarization or
self--energy effects, which are not captured by the present model.

In realistic core--shell structures, the relevant confinement levels may in fact
be quasi--bound once finite wells or barriers and further band--structure effects
are taken into account.
Accordingly, the comparison presented here is meant only to capture
order--of--magnitude estimates and qualitative trends.

\medskip
\noindent
\textbf{Interface modeling and relation to BDD.}
Band edges vary across a few atomic layers ($\sim0.3\,\mathrm{nm}$), which is thin compared with
nanocrystal length scales; see, e.g., \cite{Brus1984,Madelung2004}.
Hence we approximate a thin interfacial step by a surface $\delta$--interaction.
In BDD--type effective mass models one imposes continuity of $\psi$ and of the flux
$(1/m^\ast)\partial_n\psi$; see, for example, \cite{HarrisonQWWD4}.
Here we keep a constant $m^*$ in each calibration and absorb interfacial physics into an
\emph{effective} coupling $\alpha$; it should not be interpreted as a microscopic BDD parameter.

\begin{figure}[ht]
\centering
\begin{minipage}[b]{0.45\linewidth}
\centering
\includegraphics[width=\linewidth]{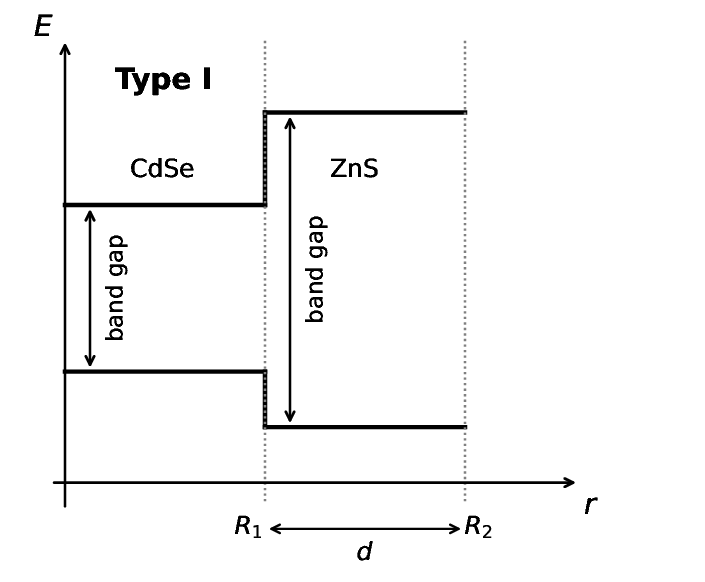}
\caption{Typical Type~I (CdSe/ZnS) band alignment.}
\label{fig:typeI_band_alignment}
\end{minipage}
\hspace{0.05\linewidth}
\begin{minipage}[b]{0.45\linewidth}
\centering
\includegraphics[width=\linewidth]{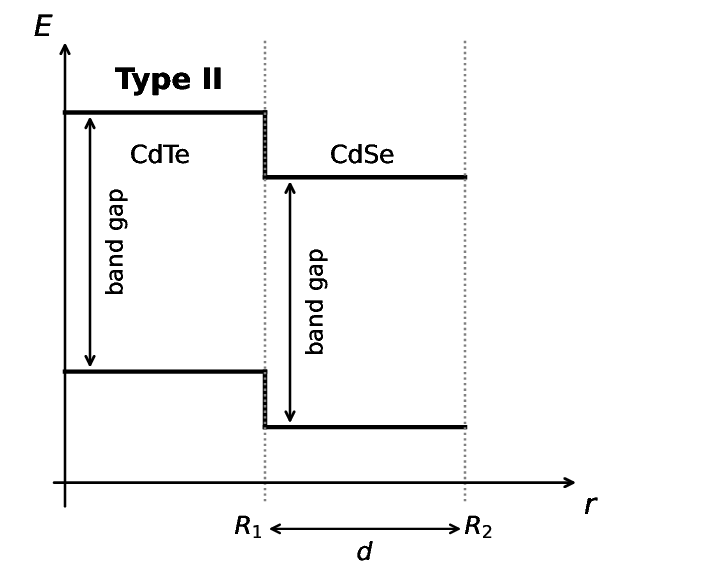}
\caption{Typical Type~II (CdTe/CdSe) band alignment.}
\label{fig:typeII_band_alignment}
\end{minipage}
\end{figure}
Figures~\ref{fig:typeI_band_alignment}--\ref{fig:typeII_band_alignment} illustrate schematic
Type~I (CdSe/ZnS) and Type~II (CdTe/CdSe) alignments; representative offsets can be found in
\cite{Li2009BandOffsets} (see also \cite{Dabbousi1997} for typical core--shell systems).
Approximating a thin step of height $\Delta V$ and width $w$ by
$\Delta V w\,\delta(|x|-R)$, we obtain the dimensionless strength
\begin{equation}\label{eq:alpha-dimless}
\alpha \simeq \frac{\Delta V}{E_0}\,\frac{w}{L_0},
\qquad
E_0=\frac{\hbar^2}{2m^*L_0^2},\quad L_0=1\,\mathrm{nm}.
\end{equation}
Thus $\mathrm{sign}(\alpha)=\mathrm{sign}(\Delta V)$.
For Type~I we use $\alpha_1<0$ (attractive inner wall) and $\alpha_2>0$ (outer barrier),
while for Type~II we reverse the signs: $\alpha_1>0$ and $\alpha_2<0$.

\medskip
\noindent
\textbf{Type~I (CdSe/ZnS): strong confinement.}
Take $m^*\simeq0.13m_0$ (CdSe), so $E_0\simeq0.29\,\mathrm{eV}$ \cite{Madelung2004}.
Choose a representative geometry $R_1=2.5\,\mathrm{nm}$, $d=1.0\,\mathrm{nm}$ ($R_2=3.5\,\mathrm{nm}$),
consistent with typical experimental ranges \cite{HinesGuyotSionnest1996,OronKazesBanin2007,Dabbousi1997}.
Fix $\alpha_1<0$ so that the \emph{single--shell} problem at $R_1$ supports an $s$--wave bound state
at a typical scale $\kappa\sim\pi/R_1$ (hence $|\alpha_1|=O(1)$).
For the outer interface, \eqref{eq:alpha-dimless} with a representative offset
(e.g.\ $\Delta V\simeq0.7\,\mathrm{eV}$ and $w\simeq0.3\,\mathrm{nm}$) gives $\alpha_2=O(1)$.
A representative set
$$
R_1=2.5,\quad R_2=3.5,\quad \alpha_1\simeq-2.5,\quad \alpha_2\simeq0.7
$$
yields an $s$--wave bound state with a confinement scale $\Delta E_{\mathrm{conf}}$
of order $10^{-1}\,\mathrm{eV}$ (in physical units), i.e.\ a strongly confined electron level.
Optical shifts are reduced by Coulomb and polarization/self--energy corrections, each typically
$\sim0.1\,\mathrm{eV}$ for radii of a few nanometers \cite{Brus1984}, so a blue shift of order
$10^{-1}\,\mathrm{eV}$ is plausible, without aiming at a quantitative fit.

\medskip
\noindent
\textbf{Type~II (CdTe/CdSe): shallow outer--shell state.}
Take $m^*\simeq0.11m_0$ (CdTe/CdSe), so $E_0\simeq0.35\,\mathrm{eV}$ \cite{Madelung2004}.
With the same geometry and the Type~II sign pattern, \eqref{eq:alpha-dimless} with a representative
downward offset (e.g.\ $\Delta V\simeq-0.8\,\mathrm{eV}$, $w\simeq0.35\,\mathrm{nm}$) gives $\alpha_2=O(1)<0$,
while a moderate inner barrier $\alpha_1=O(1)>0$ is natural.
For instance,
$$
R_1=2.5,\quad R_2=3.5,\quad \alpha_1\simeq+2.5,\quad \alpha_2\simeq-0.8
$$
typically produces (when an $s$--wave bound state exists) a very small $\kappa$ and hence a shallow level
on the meV scale, localized mainly in the outer shell.
This captures the qualitative Type~II picture: weak electron confinement.
Observed Type~II optical red shifts (often a few $10^{-1}\,\mathrm{eV}$) depend crucially on bandgap
differences and excitonic/polarization effects beyond the present effective surface--interaction model.

\medskip
\noindent
\textbf{Conclusion.}
This calibration is only a magnitude/trend check: Type~I yields a strongly confined $s$--state,
whereas the Type~II sign pattern yields a shallow outer--shell state.
Quantitative optical predictions require a multiband/exciton model and are beyond the scope of this paper.

\section{Explicit partial--wave matrices}\label{app:ml}

In this appendix we record an explicit closed form of the partial--wave matrices
$m_\ell(z)$ appearing in Lemma~\ref{lem:block-reduction}.
The formula follows from the standard spherical--wave expansion of the free
resolvent kernel.

Let $k=\sqrt z$ with $\Im k>0$.
In the present paper we evaluate the free resolvent kernel $G_z$ only at points
of the form $x=R_i\omega$ and $y=R_j\omega'$ with $i,j\in\{1,2\}$ and
$0<R_1<R_2$.
In this case the radial factors depend only on $R_{\min(i,j)}$ and
$R_{\max(i,j)}$, so the mixed terms with $i\ne j$ are symmetric in $(i,j)$.
The corresponding spherical--wave expansion reads
$$
G_z(R_i\omega,R_j\omega')
=
ik\sum_{\ell=0}^\infty\sum_{m=-\ell}^{\ell}
j_\ell\bigl(kR_{\min(i,j)}\bigr)\,
h^{(1)}_\ell\bigl(kR_{\max(i,j)}\bigr)\,
Y_{\ell m}(\omega)\,\overline{Y_{\ell m}(\omega')},
$$
for every $\omega,\omega'\in S^2$,
where $j_\ell$ and $h^{(1)}_\ell$ denote the spherical Bessel and Hankel
functions, and $\{Y_{\ell m}\}$ are the standard spherical harmonics forming
an orthonormal basis of $L^2(S^2,d\omega)$.

\begin{lem}\label{lem:ml-formula}
Assume $0<R_1<R_2$.
For $i,j\in\{1,2\}$ let $m_{ij}(z)$ be the single--layer operators with kernel
$G_z(R_i\omega,R_j\omega')$, that is,
$$
(m_{ij}(z)\varphi)(\omega)=\int_{S^2} G_z(R_i\omega,R_j\omega')\,
\varphi(\omega')\,d\omega'.
$$
Then each spherical harmonic $Y_{\ell m}$ is an eigenfunction of $m_{ij}(z)$.
More precisely,
$$
m_{ij}(z)Y_{\ell m} = \mu^{(\ell)}_{ij}(z)\,Y_{\ell m},
$$
with
$$
\mu^{(\ell)}_{11}(z)=ik\,j_\ell(kR_1)\,h^{(1)}_\ell(kR_1),\qquad
\mu^{(\ell)}_{22}(z)=ik\,j_\ell(kR_2)\,h^{(1)}_\ell(kR_2),
$$
and
$$
\mu^{(\ell)}_{12}(z)=\mu^{(\ell)}_{21}(z)
=ik\,j_\ell(kR_1)\,h^{(1)}_\ell(kR_2).
$$
Consequently, the $2\times2$ block in Lemma~\ref{lem:block-reduction} is
$$
m_\ell(z)=
\left(
\begin{array}{cc}
ik\,j_\ell(kR_1)\,h^{(1)}_\ell(kR_1) & ik\,j_\ell(kR_1)\,h^{(1)}_\ell(kR_2)\\
ik\,j_\ell(kR_1)\,h^{(1)}_\ell(kR_2) & ik\,j_\ell(kR_2)\,h^{(1)}_\ell(kR_2)
\end{array}
\right).
$$
\end{lem}



%
%
%
%

\end{document}